\documentclass[12pt,letterpaper,leqno]{article} 
\usepackage{amsmath,amsfonts,amssymb,amsthm,enumerate,multirow,array,graphicx,lscape,subcaption,color,epstopdf,setspace,mathrsfs,listings,apptools,natbib,algorithmic,algorithm,cite}
\usepackage[inline]{enumitem}
\usepackage[letterpaper,margin=1in]{geometry}
\usepackage[bookmarksdepth=2,colorlinks=false]{hyperref}
\usepackage[OT1]{fontenc}
\usepackage{titlesec,titling}
\pretitle{\begin{center}\Large\bfseries}
\posttitle{\par\end{center}}
\preauthor{\begin{center}\lineskip 0.5em\scshape\begin{tabular}[t]{c}}
\postauthor{\end{tabular}\par\end{center}}
\predate{\begin{center}\normalfont}
\postdate{\end{center}\vskip -1em}

\providecommand{\keywords}[1]{\par{\scshape Keywords:} #1}
\providecommand{\jelcodes}[1]{\par{\scshape JEL Codes:} #1}
\titleformat{\section}[hang]{\normalfont\bfseries}{\relax\thesection.~}{0pt}{}[]
\AtAppendix{\titleformat{\section}[block]{\normalfont\scshape\filcenter}{Appendix~\thesection.~}{0pt}{}[]}
\titleformat{\subsection}[runin]{\normalfont\itshape}{\thesubsection.~}{0pt}{}[.]
\AtAppendix{\titleformat{\subsection}[hang]{\normalfont\bfseries}{\relax\thesubsection.~}{0pt}{}[]}
\titleformat{\subsubsection}[runin]{\normalfont\itshape}{\thesubsubsection.~}{0pt}{}[.]
\titlespacing{\section}{\parindent}{1em}{1em}
\titlespacing{\subsection}{\parindent}{1em}{1em}
\titlespacing{\subsubsection}{\parindent}{1em}{1em}
\newtheoremstyle{imsthm}{1em}{1em}{\itshape}{\parindent}{\scshape}{.}{.5em}{}
\theoremstyle{imsthm}
\newtheorem{thm}{Theorem}
\newtheorem{lmm}[thm]{Lemma}
\newtheorem{pro}[thm]{Proposition}
\newtheorem{crl}[thm]{Corollary}
\newtheoremstyle{imsdfn}{1em}{1em}{\normalfont}{\parindent}{\scshape}{.}{.5em}{}
\theoremstyle{imsdfn}

\newtheorem*{rmk}{Remarks}

\AtAppendix{\numberwithin{equation}{section}}
\AtAppendix{\numberwithin{figure}{section}}
\AtAppendix{\numberwithin{thm}{section}}
\usepackage[font=small,labelfont=sc]{caption}
\usepackage{natbib}
\setlist{noitemsep}

\DeclareMathOperator*{\argmax}{argmax}

\usepackage{tikz}
\usetikzlibrary{calc,fit,positioning,intersections}
\usetikzlibrary{decorations.pathreplacing}
\usepackage{epigraph}

\begin{document}

\title{Costly Persuasion by a Partially Informed Sender}
\author{Shaofei Jiang\thanks{University of Bonn (email: sjiang@uni-bonn.de). I thank Sarah Auster, V. Bhaskar, Francesc Dilm\'{e}, William Fuchs, Sven Rady, Vasiliki Skreta, Maxwell Stinchcombe, Caroline Thomas, Mark Whitmeyer, Thomas Wiseman, and seminar participants at Bonn, UT Austin, DUFE-IAER, Tianjin University, SEA Annual Meeting, Midwest Economic Theory Conference, Stony Brook Game Theory Conference, AMES China, AMES, EEA-ESEM Congress, European Winter Meeting of the Econometric Society, SAET Conference for helpful comments. I gratefully acknowledge funding from the Deutsche Forschungsgemeinschaft (DFG, German Research Foundation) under Germany's Excellence Strategy - GZ 2047/1, Projekt-ID 390685813. Errors are my own.}}
\date{November 14, 2024}

\maketitle

\thispagestyle{empty}

\onehalfspacing

\begin{abstract}

I study a model of costly Bayesian persuasion by a privately and partially informed sender who conducts a public experiment. The cost of running an experiment is the expected reduction of a weighted log-likelihood ratio function of the sender's belief. This is microfounded by a Wald sequential sampling problem where good news and bad news cost differently. I focus on equilibria satisfying the D1 criterion. The equilibrium outcome depends crucially on the relative costs of drawing good and bad news in the experiment. If good news is not too costly compared to bad news, there exists a unique separating equilibrium, and the receiver learns more information thanks to sender private information. If good news is sufficiently costlier than bad news, the single-crossing property fails. There may exist pooling and partial pooling equilibria, and in some equilibria, the receiver learns less information compared to a benchmark with an uninformed sender.
\vskip 1em

\jelcodes{C72, D82, D83}
\end{abstract}

\clearpage

\pagenumbering{arabic}

\section{Introduction}

Persuasion through public experimentation is prevalent. For example, a pharmaceutical company uses clinical trials to persuade the Food and Drug Administration (FDA) to approve a new drug; an interest group funds policy pilots to lobby Congress about a policy; an entrepreneur develops a prototype to convince an investor of a new technology. Bayesian persuasion \citep{kamenica_gentzkow_2011} provides a stylized model to analyze these situations: a sender (e.g., the pharmaceutical company/the interest group/the entrepreneur) conducts a public experiment (e.g., the clinical trial/the policy pilot/the prototype) to persuade a receiver (e.g., the FDA/the congress/the investor) about a commonly unknown state of the world (e.g., the quality of the drug/the economic benefit of the policy/the reliability of the technology).

Conducting experiments are costly, and often the sender has ex ante partial private information about the state of the world. For example, the pharmaceutical company learns about the new drug through internal research and animal testings; the interest group has access to in-house experts' opinions and scholarly research on the policy; and the entrepreneur has proprietary knowledge about the new technology. Hence, the sender can choose different experiments depending on her private information, and as a result, the receiver tried to infer the state of the world using two sources of information. Aside from observing the outcome of the public experiment, he may be able to learn about the state of the world by observing what experiment the sender opts for.

Does private information undermine the effectiveness of public experimentation? A natural intuition is that the sender with more preferable private information has a stronger incentive to provide more information. For example, an entrepreneur who is more confident about the new technology will invite scrutiny from the investor by developing a more complete prototype, while a less confident entrepreneur will develop a prototype with limited features and hope that all works fine. This is in analogy to results on monotonic signaling games \citep{cho_sobel_1990} where better sender types send higher signals, and it is shown to be the case by \citet{hedlund_2017,Hedlund_2024} in a persuasion problem with sender private information where experiments are costless and either the state of the world or the sender's private information is binary. However, this implies that the the sender's choice of experiment reveals her private information, and private information leads to more informative experiments being conducted in equilibrium. Hence, the receiver should welcome sender private information. This intuition is at odds with the fact that the receiver often uses costly interventions (e.g., hiring external experts, conducting technical due diligence) to counter private information on the sender's side.

We challenge this intuition by showing that it hinges on the assumption that experiments are costless. It is not generally the case when experiments are costly for the sender. Specifically, we study a Bayesian persuasion problem where the sender is privately and partially informed about a binary state of the world. At the outset of the game, she privately observes a noisy signal, which is her \emph{type}. Sender types are ordered by their prior beliefs, so higher types are more optimistic that the state of the world is good. The receiver takes a binary action to match the state of the world, but the sender only wants to convince the receiver of the good state.

We assume that experiments are costly for the sender and focus on the log-likelihood ratio cost of experiments. That is, the cost of running an experiment equals the expected reduction of a weighted log-likelihood ratio function evaluated at the sender's belief. This is microfounded via a \citet{wald_1945} sequential sampling problem where the cost of learning good news and bad news may differ.

The assumptions that news has direct payoff consequences and that good news and bad news may have different costs are natural in many applications. For example, from the point of view of persuading the FDA, it is good news if a patient recovers, and it is bad news if a patient does not recover. For the pharmaceutical company, bad news also leads to a higher cost of conducting the clinical trial, since the company has to treat the patient using existing drugs if she does not recover. For the interest group that funds policy pilots, it needs to compensate a subject if the policy has significant negative impact on her, which leads to a higher cost. However, depending on whether the interest group is for or against the policy, showing negative impacts can be bad or good news for the purpose of lobbying. In startup funding, successfully demonstrating a feature in the prototype is good news from the point of view of persuading the investor. However, good news may be costlier if the entrepreneur uses an incentive scheme that rewards the engineers who develop the prototype with a bonus for successfully developing a feature.

Our main result is that the equilibrium outcome of the game depends on the relative costs of learning good news and bad news. As is common in signaling games, multiple equilibria exist, so we focus on equilibria that satisfy the D1 criterion \citep{banks_sobel_1987,cho_kreps_1987}. If good news is not too costly compared to bad news, it is indeed true that higher sender types have stronger preferences for more information. There exists a unique separating equilibrium, that is, the sender's choice of experiment fully reveals her type. Compared to a benchmark where the sender's type is public, every sender type bar the lowest chooses a Blackwell more informative experiment in order to deter lower types from mimicking. Hence, the sender's private information leads to the receiver learning more about the state of the world in equilibrium.

In contrast, if good news is sufficiently costlier than bad news, running a more informative experiment is punitively more expensive for higher sender types, as they are more likely to learn good news. Our model becomes a signaling game where the single-crossing property fails. Indeed, two indifference curves intersect up to \emph{three} times. The equilibrium outcome is in general not unique. In every equilibrium, an interval of sender types choose the same experiment (this includes pooling and separating equilibria as special cases), while other sender types each choose a distinct experiment. We characterize the set of pooling equilibria, that is, the receiver learns nothing from the sender's choice of experiment, and show that the sender chooses a Blackwell less informative experiment in some pooling equilibrium than in a benchmark where the sender observes no private signal. Hence, the sender's private information can lead to less information in equilibrium.

The failure of single-crossing is not specific to the log-likelihood ratio cost of experiments. If experiments have the Shannon entropy cost, we show that the sender's payoff satisfies double-crossing, and pairwise pooling is possible in equilibrium. These results suggest that costly persuasion with sender private information is a natural class of problems where the single-crossing property can fail, and the sender's incentives and equilibrium outcomes depend crucially on the cost of experiments. The fact that the sender may have single-, double-, or triple-crossing payoffs points to the richness of the model. A more general theory that goes beyond binary persuasion problems and can accommodate more cost functions may be of interest for future research but is beyond the scope of this paper.

\subsection{Related literature}

\citet{kamenica_gentzkow_2011} introduce the study of Bayesian persuasion via unrestricted, costless experiments where the sender and the receiver have common prior about the state of the world. Their main result is concavification.\footnote{The sender's payoff can be expressed as a function over the (common) posterior belief, and the sender's equilibrium payoff as a function of the prior is the concave closure of that function.} \citet{alonso_camara_2016} study an extension where the sender and the receiver have heterogeneous priors, but they agree to disagree. They derive a bijection between the sender's and the receiver's posterior beliefs, hence concavification can be applied after a translation of beliefs. \citet{gentzkow_kamenica_2014} relaxes the assumption that experiments are costless. They show that the equilibrium can be solved using concavification if the cost of experiments is posterior separable \citep{caplin_dean_leahy_2022}.

The study of Bayesian persuasion games with a privately informed sender is more recent. \citet{perez-richet_2014} studies equilibrium refinement in Bayesian persuasion where the sender is fully informed of a binary state. \citet{koessler_skreta_2023} study a more general information design problem by a fully informed designer, allowing for many agents and private messages. \citet{Zapechelnyuk_2023} studies information design by a fully informed designer who cannot choose a fully revealing experiment. All these papers assume that information transmission is costless and apply a generalized version of the inscrutability principle \citep{myerson_1983}, which says that pooling equilibria are without loss of generality. In the current paper, we assume that the public experiment is costly. Moreover, the set of available experiments does not depend on the sender's type, and the experiment outcome cannot correlate with the sender's noisy signal conditional on the state of the world, so the sender cannot verifiably disclose her private information.\footnote{If the experiment can be arbitrarily correlated with the sender's private information, the sender is able to verifiably disclose her type by publicly replicating her private signal. \citet{alonso_camara_2018} study persuasion by a partially informed sender via costless experiments and allow correlation between the public experiment and the sender's private information even conditional on the state of the world. In their paper, pooling equilibria are again without loss of generality. This idea is also explored in models of sample selection and data tampering (see, e.g., \citet{Di_Tillio_Ottaviani_Sorensen_2017,Di_Tillio_Ottaviani_Sorensen_2021,Alonso_Camara_2024}).} Therefore, the inscrutability principle does not hold.

Within the literature of Bayesian persuasion, the closest to the current paper are by \citet{hedlund_2017,Hedlund_2024}. Both papers study persuasion by a partially informed sender but assume that experiments are costless. Moreover, sender types are ordered by likelihood ratio, and the sender's payoff is monotonic in the receiver's posterior belief with respect to the likelihood ratio order. If the state of the world is binary \citep{hedlund_2017} or the sender's type is binary\footnote{When either the state or the sender's type is binary, the receiver's interim belief is one-dimensional. This is crucial to the analysis by \citet{hedlund_2017,Hedlund_2024}. Along with the aforementioned assumptions that sender types are ordered and the sender's payoff is increasing in the receiver's belief, this implies that all sender types prefer higher interim beliefs of the receiver, hence it is a monotonic signaling game \citep{cho_sobel_1990}.} and the sender's payoff function is outer concave\footnote{That is, the sender's payoff is lower from choosing the fully revealing experiment than from choosing the uninformative experiment given any common prior.} \citep{Hedlund_2024}, the single-crossing property is satisfied. Intuitively, sender types with favorable private information have stronger preferences for more informative experiments. Therefore, only separating equilibria and pooling equilibria where all sender types choose the most informative (i.e., fully revealing) experiment are selected by the D1 criterion. \citet{Kosenko_2023} points out that, under the binary state setting of \citet{hedlund_2017}, there may exist D1 pooling equilibria (that are not fully revealing) if the set of experiments is restricted or the receiver's action is discrete.\footnote{However, this is not due to failure of single-crossing. Restricting the set of available experiments does not change the sender's payoff function. If the receiver's action is binary and experiments are costless (example 2 of \citet{Kosenko_2023}), the single-crossing property is still satisfied (an observation made in Section \ref{sec4.2} of the current paper). We show in Section \ref{sec8.3} that the D1 criterion selects a continuum of pooling equilibria. Although these pooling equilibria are not fully revealing, the bad outcome reveals the bad state (i.e., they reside on the boundary of the set of experiments). This is also the case in the equilibrium presented in example 2 of \citet{Kosenko_2023}.}

The introduction of a cost of experiments substantively changes the sender's incentives. When good news is sufficiently costlier than bad news, conducting a more informative experiment is punitively more expensive for sender types with preferable private information, hence they no longer have a stronger incentive to conduct more informative experiments. In other words, the single-crossing property fails, and as a result, pooling equilibria that are not fully revealing and partial pooling equilibria may be selected by the D1 criterion.

Two other papers that feature both costly experiments and sender private information are by \citet{li_li_2013} and by \citet{degan_li_2021}. In both papers, a privately informed sender chooses from a restricted class of noisy signals that differ only on their precision, and the cost is increasing in the precision. In contrast, our paper allows the sender to choose any Blackwell experiment and assumes that the cost of a Blackwell more informative experiment is higher. There are other works that study signaling through provision of information. \citet{Bull_Watson_2019} study disclosure of hard evidence where the sender also has private soft information about a binary state. \citet{Chen_Zhang_2020} study a privately informed seller that can signal her type through both an experiment and her pricing strategy.

The failure of single-crossing is worth noting beyond the Bayesian persuasion literature, since single-crossing is widely assumed in signaling games (see, e.g., the analysis of insurance markets by \citet{rothschild_stiglitz_1976,wilson_1977}). It also plays an important role in applying various equilibrium refinements in signaling games (e.g., \citet{riley_1979,cho_kreps_1987,cho_sobel_1990,Ramey_1996}) and for monotone comparative statics \citep{milgrom_shannon_1994}. \citet{chen_ishida_suen_2022} study signaling games where the sender's preference exhibits double-crossing instead of single-crossing. They characterize equilibria satisfying the D1 criterion. In our model, when good news is much costlier than bad news, two indifference curves of different sender types can intersect three times, hence violating even the double-crossing property.

Our paper is also related to the literature on the cost of information. Posterior-separable costs have been popular in modeling attention costs (e.g., \citet{sims_1998,sims_2003}) and are used to model costs of experiments by \citet{gentzkow_kamenica_2014}. However, an experiment, as defined by \citet{blackwell_1953}, is a concept independent of beliefs, and with heterogeneous priors, it is unclear which player's beliefs should be used to compute the cost of an experiment. By studying a \citet{wald_1945} sequential sampling problem, we show that the cost of an experiment equals the expected reduction in a weighted log-likelihood ratio function evaluated at the sender's belief. This is precisely the cost function studied by \citet{pomatto_strack_tamuz_2023}, who provide an axiomatic foundation for it. Our setup of the sequential sampling problem is similar to that in \citet{brocas_carrillo_2007} and \citet{henry_ottaviani_2019}, but we assume that the cost of acquiring each signal is a random variable. Our result complements other studies which microfound costs of information through sequential information acquisition (e.g., \citet{morris_strack_2019,bloedel_zhong_2020,hebert_woodford_2023}).

\vskip 1em

The rest of the paper is organized as follows. Section \ref{sec2} presents the model. Section \ref{sec4} shows that the single-crossing property fails when good news is sufficiently costlier than bad news. In Section \ref{sec5}, we characterize the set of D1 pooling equilibrium outcomes when single-crossing fails and the unique separating equilibrium outcome under single-crossing. In Sections \ref{sec8}, we discuss partial pooling equilibria. We also present a sequential sampling problem which microfounds the log-likelihood ratio cost function, and we show that single-crossing fails under the Shannon entropy cost. All proofs are in the Appendix.

\section{The Model}\label{sec2}

There is a sender (she), and a receiver (he). At the outset of the game, Nature determines a binary state of the world \(\omega\in\Omega:=\{G,B\}\) and a signal \(\theta\in\Theta:=\{1,2,\dots,N\}\) according to a commonly known joint distribution \(F\) with full support over \(\Omega\times\Theta\). Let \(\mu_0\) be the probability of the good state (i.e., \(\omega=G\)), and \(\mu_\theta\) the probability of the good state conditional on the signal realization \(\theta\). We assume that \(0<\mu_1<\mu_2<\dots<\mu_N<1\). The sender privately observes the signal \(\theta\), and neither player observes the state \(\omega\). Therefore, \(\theta\) is the sender's \emph{type}, and the sender's prior belief on the good state is \(\mu_\theta\). On the other hand, the receiver's prior belief  is \(\mu_0\).

The game proceeds as follows. The sender publicly chooses an experiment \(\pi\) on the state of the world. The experiment yields a binary outcome \(s\in\{g,b\}\).\footnote{It is without loss of generality to focus on binary experiments given the assumptions that Blackwell more informative experiments are costlier and that the receiver follows a threshold decision rule. These assumptions are presented in Sections \ref{sec2.3} and \ref{sec2.4}, respectively.} That is, \(\pi:\Omega\to\Delta(\{g,b\})\). The outcome of the chosen experiment \(s\) is determined according to the distribution \(\pi(\cdot|\omega)\) and is publicly observed. Notice that it is conditionally independent of the sender's private information. The receiver takes a binary action \(a\in\{0,1\}\), and payoffs are realized.\footnote{The model can be equivalently formulated as a mechanism design problem with an informed principal, where the outcome \(g\) is an action recommendation for the receiver to take the action \(a=1\), and the outcome \(b\) is an action recommendation for the receiver to take the action \(a=0\).}

\subsection{Strategies}

Given an experiment \(\pi\), let \(p=\pi(g|G)\) and \(q=\pi(g|B)\). Without loss of generality, \(p\geq q\). An experiment is thus identified with the pair of probabilities \((p,q)\), and the set of feasible experiments is \(\Pi=\{(p,q):1\geq p\geq q\geq 0\}\). We denote by \(\Pi^\circ\) the interior of \(\Pi\). An experiment \((p,q)\) is Blackwell more informative than another experiment \((p',q')\) if and only if \(\frac{q}{p}\leq\frac{q'}{p'}\) and \(\frac{1-p}{1-q}\leq\frac{1-p'}{1-q'}\).

A pure strategy of the sender \(\{\pi_\theta\}_{\theta\in\Theta}\) is the collection of experiments chosen by all sender types, where \(\pi_\theta\in\Pi\) is the experiment chosen by the type \(\theta\) sender. A pure strategy of the receiver is \(\mathbf{a}:\Pi\times\{g,b\}\to\{0,1\}\). It selects an action at every information set of the receiver, which is identified by the sender's choice of experiment \(\pi\) and its outcome \(s\).

\subsection{Beliefs}

After observing the sender's choice of experiment but before seeing its outcome, the receiver forms a belief about the sender's type and the state of the world. Let \(\gamma(\theta|\pi)\) denote his belief that the sender's type is \(\theta\) after experiment \(\pi\) is chosen. Then \(\beta(\pi):=\sum_{\theta\in\Theta}\gamma(\theta|\pi)\mu_\theta\) is the receiver's \emph{interim belief} on the good state. Notice that \(\beta(\pi)\in[\mu_1,\mu_N]\).

After the outcome is observed, both players update their beliefs. Let \(\hat{\mu}(\theta,\pi,s)\) and \(\hat{\beta}(\pi,s)\) be the posterior beliefs of the type \(\theta\) sender and the receiver, respectively, that the state is good after observing outcome \(s\) from experiment \(\pi\).

\subsection{Cost of experiments and the sender's payoff}\label{sec2.3}

The sender strictly prefers the high receiver action over the low action. Her payoff \(v(a,\pi|\theta)=a-c(\pi|\mu_\theta)\) consists of two parts: a reward which is normalized to 1 if the receiver chooses the high action, minus the cost of the experiment \(c(\pi|\mu_\theta)\), which equals the expected reduction of a weighted log-likelihood ratio function evaluated at her belief. That is, the cost of running an experiment \(\pi\) given the sender's prior belief \(\mu\) is
\begin{equation*}
	c(\pi|\mu) = \mathbb{E}[H(\mu)-H(\hat{\mu})],
\end{equation*}
where
\begin{equation*}
	H(\mu) = C_g\mu\ln\left(\frac{1-\mu}{\mu}\right)+C_b(1-\mu)\ln\left(\frac{\mu}{1-\mu}\right),
\end{equation*}
\(C_g,C_b>0\), and \(\hat{\mu}\) is the sender's posterior belief induced by the experiment \(\pi\).\footnote{The cost \(c(\pi|\mu)\in\mathbb{R}_+\cup\{+\infty\}\) is well-defined for all \(\mu\in(0,1)\) and \(\pi\neq(0,0),(1,1)\). For completeness, let \(c((0,0)|\mu)=c((1,1)|\mu)=0\) for all \(\mu\in(0,1)\).}

A few remarks are in order regarding the log-likelihood ratio cost function. First, \citet{pomatto_strack_tamuz_2023} show that this is the only family of cost functions satisfying three axioms.\footnote{The three axioms are: first, a Blackwell more informative experiment is costlier; second, the cost of generating independent experiments is the sum of their individual costs; third, the cost of generating an experiment with some probability is linear in the probability (see Theorem 5 in the online appendix of \citet{pomatto_strack_tamuz_2023}). For more than two states, a continuity condition is needed.} We microfound the cost function in Section \ref{sec7} and show that the parameters \(C_g\) and \(C_b\) are the costs of drawing good and bad news, respectively. Hence, the parameterization \(C_g<C_b\) models scenarios where bad news is costlier (e.g., pharmaceutical companies conducting clinical trials), while \(C_g>C_b\) models scenarios where good news is costlier (e.g., entrepreneurs developing prototypes). Second, the cost of an experiment \(c(\pi|\mu)\) depends on the sender's prior belief \(\mu\). Intuitively, a more optimistic sender is more likely to learn good news, and a more pessimistic sender is more likely to learn bad news. Hence, depending on the relative costs of learning good new and bad news, they can have different costs of running the same experiment. We show in the proof of Proposition \ref{pro3} that the cost of a given experiment is an increasing (decreasing) affine function of the sender's prior if \(\frac{C_g}{C_b}\) is above (below) a threshold. Finally, the cost of a Blackwell more informative experiment is always higher. The uninformative experiment\footnote{An experiment is uninformative if \(p=q\). Since all uninformative experiments are Blackwell equivalent and have zero cost, we will identify and refer to them as ``the uninformative experiment.''} has zero cost, and any experiment that can reveal the state (i.e., \(q<p=1\) or \(0=q<p\)) has an infinite cost.

\subsection{The receiver's payoff}\label{sec2.4}

The receiver's payoff \(u(a,\omega)\) depends on his action \(a\) and the state of the world \(\omega\). By normalization, \(u(0,G)=u(0,B)=0\), \(u(1,G)=1\), and \(u(1,B)=-\bar{\beta}/(1-\bar{\beta})\), hence the receiver follows a threshold decision rule and takes the high action if and only if his posterior belief is at least \(\bar{\beta}\).\footnote{As is standard in Bayesian persuasion, we assume that the receiver takes the sender preferred action when he is indifferent.} We assume that \(\mu_N<\bar{\beta}\), so the receiver is never persuaded at the interim stage. The sender's private signal is relatively noisy. Even if the receiver knows that the sender is the highest type, it is not optimal for the receiver to take the high action without learning from an experiment. We relax this assumption in Section \ref{sec5.2} and give a partial characterization of equilibrium outcomes when \(\mu_N>\bar{\beta}\).

\subsection{Equilibrium}

An \emph{equilibrium} consists of pure strategies of the players, \(\{\pi_\theta\}_{\theta\in\Theta}\) and \(\mathbf{a}:\Pi\times\{g,b\}\to\{0,1\}\), and the receiver's system of beliefs \(\beta:\Pi\to[\mu_1,\mu_N]\) and \(\hat{\beta}:\Pi\times\{g,b\}\to[0,1]\), such that
\begin{enumerate}[label=(\arabic*)]
	\item Given the receiver's strategy \(\mathbf{a}\), the sender's strategy is optimal, i.e.,
	\begin{equation*}
		\pi_\theta\in\argmax_{\pi\in\Pi}\mathbb{E}[v(\mathbf{a}(\pi,s),\pi|\theta)]
	\end{equation*}
	for all \(\theta\in\Theta\);
	\item The receiver is sequentially rational, i.e., \(\mathbf{a}(\pi,s)=1\) if and only if \(\hat{\beta}(\pi,s)\geq\bar{\beta}\);
	\item Beliefs are updated using Bayes' rule whenever possible. That is,
	\begin{equation*}
		\beta(\pi)=\frac{F(\{G\}\times\Theta_\pi)}{F(\Omega\times\Theta_\pi)}
	\end{equation*}
	if \(\Theta_\pi:=\{\theta:\pi_\theta=\pi\}\) is nonempty, and
	\begin{equation*}
		\hat{\beta}(\pi,s)=\mathbf{B}(\beta(\pi),\pi,s):=\frac{\beta(\pi)\pi(s|G)}{\beta(\pi)\pi(s|G)+(1-\beta(\pi))\pi(s|B)}
	\end{equation*}
	if \(\pi(s|G)+\pi(s|B)\neq 0\).
\end{enumerate}
We say an equilibrium is a \emph{persuasion equilibrium} if the high action is taken with positive probability on the equilibrium path. Otherwise, it is an \emph{uninformative equilibrium}. A persuasion equilibrium is a \emph{pooling equilibrium} if all sender types choose the same experiment. It is a \emph{separating equilibrium} if every sender type chooses a different experiment. A persuasion equilibrium that is neither pooling nor separating is a \emph{partial pooling equilibrium}. Given an equilibrium, we call the collection of experiments \(\{\pi_\theta\}_{\theta\in\Theta}\) the \emph{equilibrium outcome}.

\subsection{The D1 criterion}

We focus on equilibria that satisfy the D1 criterion. Given the receiver's interim belief \(\beta\in[\mu_1,\mu_N]\), let \(\bar{v}(\beta,\pi|\theta)\) be the type \(\theta\) sender's expected payoff from choosing an experiment \(\pi\). That is,
\begin{equation*}
	\bar{v}(\beta,\pi|\theta) = \mathbb{P}[\mathbf{B}(\beta,\pi,s)\geq\bar{\beta}]-c(\pi|\mu_\theta).
\end{equation*}
Fixing an equilibrium, let \(v_\theta^\star\) be the type \(\theta\) sender's equilibrium payoff, and for any deviation \(\pi\in\Pi\setminus\{\pi_\theta\}_{\theta\in\Theta}\), let
\begin{gather*}
	D_\theta(\pi) = \{\beta\in[\mu_1,\mu_N]:\bar{v}(\beta,\pi|\theta)>v_\theta^\star\}, \\
	D^0_\theta(\pi) = \{\beta\in[\mu_1,\mu_N]:\bar{v}(\beta,\pi|\theta)\geq v_\theta^\star\}.
\end{gather*}
\(D_\theta(\pi)\) is the set of the receiver's interim beliefs given which \(\pi\) is a profitable deviation for the type \(\theta\) sender, and \(D^0_\theta(\pi)\) is the set of the receiver's interim beliefs given which the deviation gives the sender at least the same payoff as her equilibrium payoff. An important special case is when \(D^0_\theta(\pi)\) is empty, that is, the deviation \(\pi\) is strictly equilibrium dominated for the sender regardless of the receiver's interim belief.

An equilibrium satisfies the D1 criterion if: for all pairs of distinct sender types \((i,j)\) and deviations \(\pi\), if \(D_i^0(\pi)\subsetneq D_j(\pi)\), then the receiver's off-path interim belief \(\beta(\pi)\) lies in the convex hull of \(\{\mu_\theta:\theta\neq i\}\). In words, if some sender type \(j\) is keener to deviate to \(\pi\) than sender type \(i\), in the sense that this deviation is profitable for her given a larger set of receiver beliefs than for sender type \(i\), then the receiver should not attribute this deviation to sender type \(i\). A special case is when \(\pi\) is strictly equilibrium dominated for all sender types \(\theta<n\). Then the D1 criterion requires that \(\beta(\pi)\geq \mu_n\) if \(D_n(\pi)\) is nonempty.

\section{The Single-Crossing Property}\label{sec4}

The sender faces a trade-off between information quality and cost. She wants to minimize the cost, but the experiment must be sufficiently informative so that the receiver is willing to take the high action if the outcome turns out good. Given the receiver's interim belief \(\beta\), an experiment \(\pi=(p,q)\) is said to be \emph{persuasive at belief \(\beta\)} if the receiver will take the high action following the good outcome, that is, if \(\mathbf{B}(\beta,\pi,g)\geq\bar{\beta}\). Equivalently,
\begin{equation*}
	\frac{q}{p}\leq \mathbf{Q}(\beta):=\frac{\beta}{1-\beta}\Big/\frac{\bar{\beta}}{1-\bar{\beta}}.
\end{equation*}
The sender's expected payoff from choosing a persuasive experiment \(\pi\) is
\begin{equation*}
	f(\pi,\mu_\theta) := \mu_\theta p+(1-\mu_\theta)q - c(\pi|\mu_\theta).
\end{equation*}
All other experiments are \emph{unpersuasive} (at belief \(\beta\)). If an unpersuasive experiment is chosen, the receiver takes the low action regardless of the experiment's outcome, and the sender's expected payoff is \(-c(\pi|\mu_r)\leq 0\). The sender's payoff is zero if and only if she chooses the uninformative experiment.

\subsection{Marginal rate of substitution and single-crossing}\label{sec4.1}

In an equilibrium, the sender chooses either the uninformative experiment or some experiment \(\pi\) that is persuasive at the receiver's interim belief \(\beta(\pi)\). In the latter case, her expected payoff is \(f(\pi,\mu_\theta)\). The marginal rate of substitution (of \(p\) for \(q\) for the sender type \(\theta\))
\begin{equation*}
	MRS(\pi|\mu_\theta) = -\frac{\partial f(\pi,\mu_\theta)/\partial p}{\partial f(\pi,\mu_\theta)/\partial q}
\end{equation*}
is the marginal utility of increasing \(p\) relative to increasing \(q\). Increasing either \(p\) or \(q\) leads to a higher probability of the high receiver action, but increasing \(p\) increases the cost, whereas increasing \(q\) reduces the cost. Hence, the marginal rate of substitution is negative when \(p\) is small, but it becomes positive and goes to infinity as \(p\) goes to \(1\).

The marginal rate of substitution has the usual geometric representation. In Figure \ref{fig.2}, the right triangle is the set of experiments \(\Pi\), where \(p\) and \(q\) are shown on the horizontal and vertical axes, respectively. The solid curve shows an indifference curve of some sender type \(h\). The marginal rate of substitution \(MRS(\pi|\mu_h)\) is the slope of the indifference curve at the experiment \(\pi\).

\begin{figure}[t]
	\centering
	\begin{tikzpicture}[scale=.65]
        \node[inner sep=0pt] at (5,5) {\includegraphics[width=.405\textwidth]{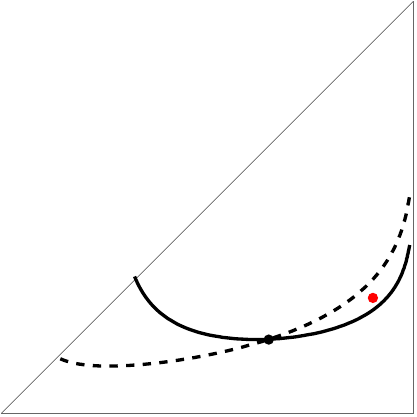}};
        \node[below] at (5,0) {$p$};
        \node[right] at (10,5) {$q$};
        \node[above] at (6.5,1.8) {$\pi$};
        \draw (9.1,2.8) -- (11,2.8) node[right] {$\tilde{\pi}$};
        \draw[very thick] (0,9) -- (1,9) node[right] {high sender type \(h\)};
        \draw[very thick,dashed] (0,8) -- (1,8) node[right] {low sender type \(l\)};
        \node at (.5,10) {\(\uparrow\)};
        \node[right] at (1,10) {direction of increasing payoffs};
    \end{tikzpicture}
	\caption{The single-crossing property and an example of a nearby deviation}\label{fig.2}
\end{figure}
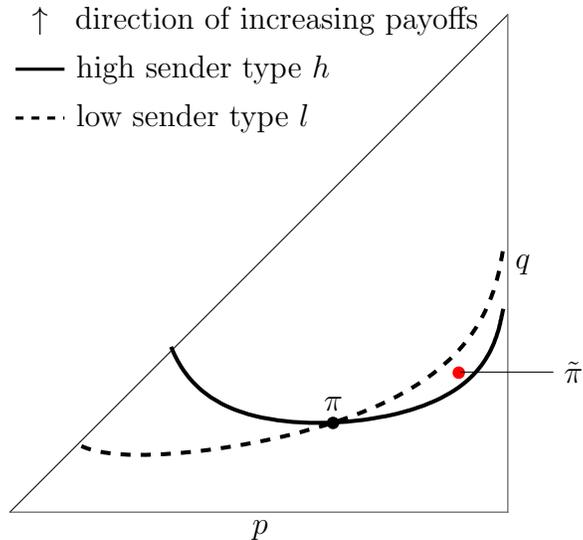

At the uninformative experiment, the marginal rate of substitution
\begin{equation}\label{eq.3}
	MRS(\pi|\mu_\theta) = -\frac{\mu_\theta}{1-\mu_\theta}
\end{equation}
is decreasing in \(\mu_\theta\). We say that \emph{the single-crossing property} is satisfied if the marginal rate of substitution is weakly decreasing in the sender's prior belief at all experiments, i.e.,
\begin{equation*}
	\frac{\partial}{\partial\mu}MRS(\pi|\mu)\leq 0
\end{equation*}
for all \(\pi\in\Pi^\circ\).\footnote{The marginal rate of substitution is not well defined if \(p=1\) or \(q=0\), hence the need to define single-crossing only on \(\Pi^\circ\).} Geometrically, this implies that the indifference curve of a lower sender type is everywhere steeper than that of a higher sender type, hence any two indifference curves intersect at most once. For example, the dashed curve in Figure \ref{fig.2} shows the indifference curve of a lower sender type \(l<h\) through \(\pi\). It has a more upward slope than the indifference curve of the high type sender \(h\) at \(\pi\), and the two indifference curves intersect only at \(\pi\).

When the single-crossing property holds, higher types of the sender have stronger incentives to conduct more informative experiments. That is, given any experiment, there exists a more informative experiment that decreases all but the highest type sender's payoff. Hence, the highest type of the sender can credibly prove her type by providing more information. This is formalized in Section \ref{sec4.3}.

\subsection{Failure of single-crossing}\label{sec4.2}

If experiments are costless, the single-crossing property is always satisfied, since the sender's payoff \(f(\pi,\mu_\theta)=\mu_\theta p+(1-\mu_\theta)q\) is simply the probability of the good outcome, and the marginal rate of substitution, given by (\ref{eq.3}), is decreasing in \(\mu_\theta\). When experiments are costly, the sender's trade-off is more involved, as the sender's choice of experiment affects her payoff also through its cost. Proposition \ref{pro3} below shows that the single-crossing property fails if good news is sufficiently costlier than bad news. Intuitively, higher types of the sender are more likely to learn good news. Hence, costlier good news decreases their incentives to provide more information, and when good news is sufficiently costlier than bad news, it is no longer true that higher sender types have stronger preferences for more informative experiments.

\begin{pro}\label{pro3}
	There exists \(\hat{K}:\mathbb{R}_{++}\to\mathbb{R}_{++}\) such that the single-crossing property is satisfied if and only if \(C_g\leq\hat{K}(C_b)\). \(\hat{K}\) is twice continuously differentiable, increasing and concave, and for all \(C_b>0\), \(\hat{K}(C_b)>C_b\), and \(\lim_{C_b\downarrow 0}\hat{K}'(C_b)=\infty\).
\end{pro}

When single-crossing fails, there exist two loci of experiments at which the indifference curves of all sender types are tangent. These experiments are highlighted by the curves \(p=\hat{\mathbf{p}}(q)\) and \(p=\check{\mathbf{p}}(q)\) in Figure \ref{fig.4}. If an experiment is sufficiently uninformative (i.e., if \(p<\hat{\mathbf{p}}(q)\)) or sufficiently informative (i.e., if \(p>\check{\mathbf{p}}(q)\)), the marginal rate of substitution is decreasing in the sender's type. In the intermediate region (i.e., if \(\hat{\mathbf{p}}(q)<p<\check{\mathbf{p}}(q)\)), the marginal rate of substitution is increasing in the sender's type. Hence, two indifference curves can intersect up to \emph{three} times, once in each region. 

Consider two experiments \(\pi_1=(p_1,q_1)\) and \(\pi_2=(p_2,q_2)\) such that \(p_1=\hat{\mathbf{p}}(q_1)\) and \(p_2=\check{\mathbf{p}}(q_2)\). Figure \ref{fig.4} shows the indifference curves of two types of the sender \(h>l\) at these experiments. At experiment \(\pi_1\), the high-type sender's indifference curve is more convex than the low-type sender's indifference curve. Hence, in a neighborhood of \(\pi_1\), the high-type sender's indifference curve is higher than the low-type sender's indifference curve. But at experiment \(\pi_2\), the low-type sender's indifference curve is more convex and therefore higher than that of the high-type sender.

Proposition \ref{pro3_part2} below summarizes the results.

\begin{pro}\label{pro3_part2}
	If \(C_g>\hat{K}(C_b)\), there exist \(\hat{\mathbf{p}},\check{\mathbf{p}}:(0,1)\to(0,1)\) such that \(q<\hat{\mathbf{p}}(q)<\check{\mathbf{p}}(q)\) for all \(q\in(0,1)\), and
	\begin{equation}\label{eq.4}
		\frac{\partial}{\partial\mu}MRS(\pi|\mu) \left\{\begin{array}{ll}
			<0 & \text{if \(p<\hat{\mathbf{p}}(q)\) or \(p>\check{\mathbf{p}}(q)\)} \\
			=0 & \text{if \(p=\hat{\mathbf{p}}(q)\) or \(p=\check{\mathbf{p}}(q)\)} \\
			>0 & \text{if \(\hat{\mathbf{p}}(q)<p<\check{\mathbf{p}}(q)\)}
		\end{array}\right..
	\end{equation}
	Moreover, \(\hat{\mathbf{p}}(q)\) is decreasing in \(C_g\) and increasing in \(C_b\), and \(\check{\mathbf{p}}(q)\) is increasing in \(C_g\) and decreasing in \(C_b\) for all \(q\in(0,1)\).
\end{pro}

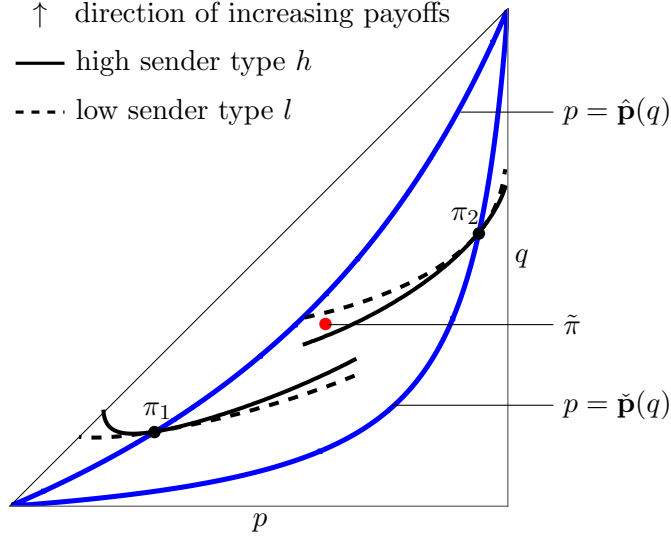
\begin{figure}[t]
	\centering
	\begin{tikzpicture}[scale=.65]
        \small
        \node[inner sep=0pt] at (5,5) {\includegraphics[width=.405\textwidth]{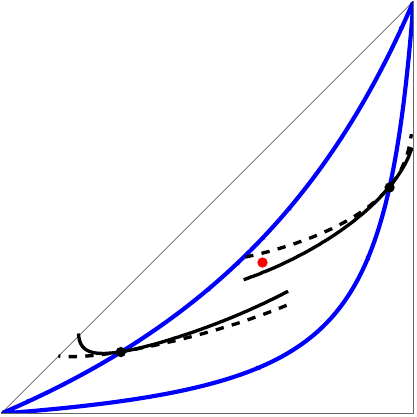}};
        \node[below] at (5,0) {$p$};
        \node[right] at (10,5) {$q$};
        \draw (9.1,8) -- (11,8) node[right] {$p=\hat{\mathbf{p}}(q)$};
        \draw (7.7,2) -- (11,2) node[right] {$p=\check{\mathbf{p}}(q)$};
        \node[above] at (2.9,1.5) {$\pi_1$};
        \node[above] at (9.2,5.5) {$\pi_2$};
        \draw (6.4,3.65) -- (11,3.65) node[right] {$\tilde{\pi}$};
        \draw[very thick] (0,9) -- (1,9) node[right] {high sender type \(h\)};
        \draw[very thick,dashed] (0,8) -- (1,8) node[right] {low sender type \(l\)};
        \node at (.5,10) {\(\uparrow\)};
        \node[right] at (1,10) {direction of increasing payoffs};
    \end{tikzpicture}
	\caption{Experiments that are robust to nearby deviations when the single-crossing property fails}\label{fig.4}
\end{figure}

\subsection{Nearby deviations}\label{sec4.3}

The marginal rate of substitution calibrates the sender's incentive to make slight deviations. If two sender types have different marginal rates of substitution at an experiment, the D1 criterion rules out equilibria where they both choose this experiment.

\begin{lmm}\label{lmm4}
	Let \(\pi\in\Pi^\circ\), and \(i,j\in\Theta\). If \(MRS(\pi|\mu_i)\neq MRS(\pi|\mu_j)\), any equilibrium such that \(\pi_i=\pi_j\) does not satisfy the D1 criterion.
\end{lmm}

We illustrate Lemma \ref{lmm4} using Figure \ref{fig.2}. For simplicity, assume that there are only two sender types \(l<h\), and there is a pooling equilibrium where they both choose \(\pi\). The experiment \(\pi\) is persuasive at \(\mu_0\), as it is the receiver's interim belief in the pooling equilibrium. Hence, there is a \emph{nearby deviation} \(\tilde{\pi}\) that is persuasive at the receiver's highest possible interim belief \(\mu_h\). For the high-type sender, deviating to \(\tilde{\pi}\) is profitable if the receiver's interim belief is sufficiently high. Specifically, if \(\beta(\tilde{\pi})=\mu_h\), the sender gets \(f(\tilde{\pi},\mu_h)\), which is higher than her equilibrium payoff \(f(\pi,\mu_h)\). In contrast, this deviation is strictly equilibrium dominated for the low-type sender. Her payoff from deviating is at most \(f(\tilde{\pi},\mu_l)\), which is strictly less than her equilibrium payoff \(f(\pi,\mu_l)\). Therefore, the D1 criterion asserts that it is only reasonable to attribute the deviation to the high type, that is, \(\beta(\tilde{\pi})=\mu_h\). But this makes \(\tilde{\pi}\) a profitable deviation for the high type. Therefore, there is no D1 pooling equilibrium where both types choose \(\pi\).

As a result of Lemma \ref{lmm4}, when the single-crossing property is satisfied, no pooling or partial pooling equilibrium satisfies the D1 criterion. In any D1 equilibrium, except for those who choose the uninformative experiment, all sender types choose distinct experiments. When the single-crossing property fails, pooling or partial pooling equilibria may be selected by the D1 criterion. We say an experiment \(\pi=(p,q)\) is \emph{robust to nearby deviations} if \(p=\hat{\mathbf{p}}(q)\). For example, at experiment \(\pi_1\) in Figure \ref{fig.4}, the high type sender's indifference curve is lower than the low type sender's indifference curve, hence any nearby deviation that is profitable for the high-type sender is also profitable for the low-type sender. Therefore, the D1 criterion is moot in determining the receiver's beliefs after seeing these nearby deviations. In contrast, if \(p=\check{\mathbf{p}}(q)\) (e.g., at experiment \(\pi_2\)), the high-type sender's indifference curve is lower. Therefore, we can find a nearby deviation (e.g., \(\tilde{\pi}\)) that is profitable only for the high-type sender but strictly equilibrium dominated for the low-type sender. This eliminates \(\pi_2\) as a pooling D1 equilibrium outcome.

\subsection{Large deviations}\label{sec4.4}

A pooling equilibrium outcome that is robust to nearby deviations does not necessarily satisfy the D1 criterion. We will illustrate in Figure \ref{fig.6} that there may be a \emph{large deviation} that is profitable only for the high-type sender.

\begin{figure}[t]
	\centering
	\begin{tikzpicture}[scale=.65]
        \small
        \node[inner sep=0pt] at (5,5) {\includegraphics[width=.405\textwidth]{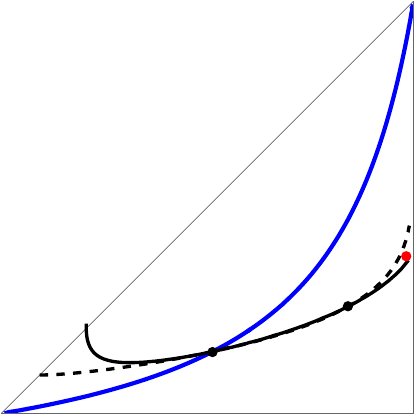}};
        \node[below] at (5,0) {$p$};
        \node[right] at (10,5) {$q$};
        \draw (9.6,8) -- (11,8) node[right] {$p=\hat{\mathbf{p}}(q)$};
        \draw (9.9,3.8) -- (11,3.8) node[right] {$\tilde{\pi}$};
        \node[above] at (5.1,1.5) {$\pi$};
        \node[above] at (8.4,2.6) {$\pi^\prime$};
        \draw[very thick] (0,9) -- (1,9) node[right] {high sender type \(h\)};
        \draw[very thick,dashed] (0,8) -- (1,8) node[right] {low sender type \(l\)};
        \node at (.5,10) {\(\uparrow\)};
        \node[right] at (1,10) {direction of increasing payoffs};
    \end{tikzpicture}
	\caption{An example of a large deviation that breaks a pooling equilibrium}\label{fig.6}
\end{figure}
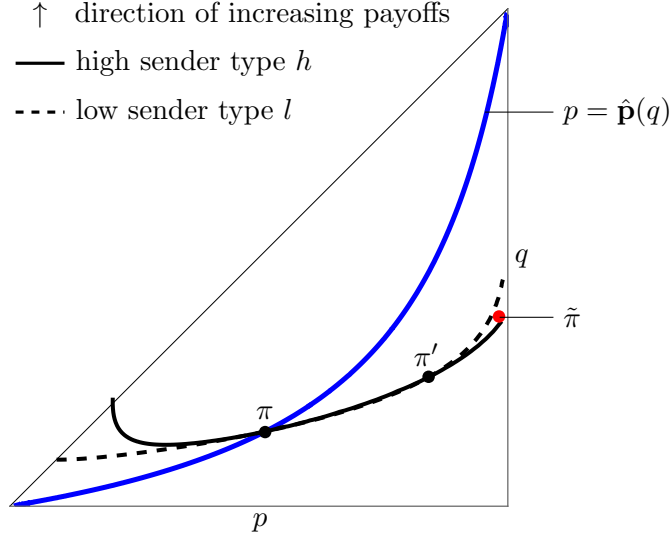

Let \(\pi=(\hat{\mathbf{p}}(q),q)\) be a pooling equilibrium outcome. The indifference curves of the two sender types are tangent at \(\pi\), but they intersect again to the right of \(\pi\), at \(\pi^\prime\), and the low-type sender's indifference curve is steeper at \(\pi^\prime\). Therefore, we can find a deviation close to \(\pi^\prime\), say \(\tilde{\pi}\), that is above the high-type sender's indifference curve and below the low-type sender's indifference curve. Two cases are possible. If \(\tilde{\pi}\) is persuasive at belief \(\mu_h\), then this deviation is profitable for the high-type sender for some interim belief of the receiver, but it is strictly equilibrium dominated for the low-type sender. In this case, the D1 criterion requires that the receiver attributes this deviation to the high sender type, thus breaking the pooling equilibrium. However, if \(\tilde{\pi}\) is unpersuasive even at belief \(\mu_h\)--that is, the receiver always chooses the low action--then \(\tilde{\pi}\) is not a profitable deviation for either type of the sender. In this case, the D1 criterion does not impose any restriction on the receiver's interim belief and cannot rule out the pooling equilibrium outcome \(\pi\).

We say an equilibrium \(\pi=(\hat{\mathbf{p}}(q),q)\) is \emph{robust to large deviations} if, for all experiments \(\tilde{\pi}\) that are persuasive at belief \(\mu_N\), \(f(\tilde{\pi},\mu_N)>f(\pi,\mu_N)\) implies \(f(\tilde{\pi},\mu_{\theta})>f(\pi,\mu_{\theta})\) for all \(\theta\in\Theta\). In words, if a deviation is profitable for the highest sender type, it is profitable for \emph{all lower} sender types, so the highest sender type cannot convince the receiver of her type by choosing an off the equilibrium path experiment.

Intuitively, this requires that \(\pi^\prime\) in Figure \ref{fig.6} be unpersuasive at any belief lower than \(\mu_N\). The following lemma shows that, as \(\pi\) moves higher along the curve \(p=\hat{\mathbf{p}}(q)\), \(\pi^\prime\) is also higher. Hence, restricting to experiments that are robust to nearby deviations, robustness to large deviations is equivalent to a lower bound on \(q\).

\begin{lmm}\label{lmm.a6}
	Let \(C_g>\hat{K}(C_b)\). There exists \(\hat{q}\geq 0\) such that an experiment \((\hat{\mathbf{p}}(q),q)\) that is persuasive at belief \(\mu_N\) is robust to large deviations if and only if \(q\geq\hat{q}\).
\end{lmm}

\section{Pooling and Separating Equilibria}\label{sec5}

We characterize in this section equilibria satisfying the D1 criterion. Proposition \ref{pro5} characterizes the set of all D1 pooling equilibrium outcomes when the single-crossing property fails, and Proposition \ref{pro6} characterizes the unique equilibrium outcome when the single-crossing property holds.

It is useful to define
\begin{equation*}
	V(\mu,\beta) := \sup\{f(\pi,\mu):~\pi~\text{is persuasive at belief}~\beta\}.
\end{equation*}
Intuitively, this is the sender's equilibrium payoff in a symmetric information version of the game where the sender and the receiver have prior beliefs \(\mu\) and \(\beta\), respectively, and they agree to disagree.\footnote{That is, there is no sender private information. The sender publicly chooses an experiment, and the receiver takes an action having observed the outcome of that experiment. The receiver learns about the state of the world only from the experiment's outcome and, in equilibrium, updates her belief from her prior using Bayes' rule.} Since \(f(\pi,\mu)\) is strictly convex in \(\pi\), if \(V(\mu,\beta)>0\), the supremum is uniquely obtained at some experiment \(\pi\); if \(V(\mu,\beta)=0\), \(f(\pi,\mu)<0\) for all experiments \(\pi\) that are persuasive at belief \(\beta\), and the sender obtains zero payoff only by choosing the uninformative experiment.\footnote{Notice that \(\lim_{p\downarrow 0}f((p,\mathbf{Q}(\beta)p),\mu)=0\). Therefore, \(V(\mu,\beta)\) is nonnegative.}

In our model with sender private information, \(V(\mu_\theta,\mu_1)\) is the sender's \emph{payoff guarantee}. Since the receiver's interim belief is at least \(\mu_1\), the sender gets \(f(\pi,\mu_\theta)\) regardless of the receiver's interim belief from deviating to an experiment \(\pi\) that is persuasive at belief \(\mu_1\). Therefore, there is no profitable deviation only if her equilibrium payoff is at least \(V(\mu_\theta,\mu_1)\).

The notation \(V(\mu,\beta)\) also provides a convenient way to describe the cost of experiments. In Lemma \ref{lmm.a1}, we show that \(V(\mu,\beta)>0\) if and only if a convex combination of \(C_g\) and \(C_b\) is less than some threshold, that is, the ``average'' cost of experimenting is sufficiently low.

\subsection{Pooling equilibria}

A pooling equilibrium outcome is selected by the D1 criterion if and only if it is robust both to nearby deviations and to large deviations. The following proposition characterizes the set of all D1 pooling equilibrium outcomes.

\begin{pro}\label{pro5}
	Let \(C_g>\hat{K}(C_b)\). An experiment \(\pi=(p,q)\) is a D1 pooling equilibrium outcome if and only if:
	\begin{enumerate}[label=(\roman*)]
		\item The receiver's obedience constraint: \(\frac{q}{p}\leq\mathbf{Q}(\mu_0)\);
		\item The lowest sender type's participation constraint: \(f(\pi,\mu_1)\geq V(\mu_1,\mu_1)\);
		\item Robustness to nearby deviations: \(p=\hat{\mathbf{p}}(q)\);
		\item Robustness to large deviations: \(q\geq\hat{q}\), where \(\hat{q}\) is defined in Lemma \ref{lmm.a6}.
	\end{enumerate}
\end{pro}

The receiver's obedience constraint (i) says that the experiment \(\pi\) must be persuasive at the receiver's interim belief \(\mu_0\), so the receiver is willing to take the high action following the good outcome. The participation constraint (ii) says that the lowest sender type's equilibrium payoff must be at least her payoff guarantee. Conditions (iii) and (iv), as we show earlier, ensure that the equilibrium outcome is robust to nearby deviations and large deviations. Restricting to experiments that are robust to nearby deviations, the participation constraint of the lowest sender type implies those of all sender types. This greatly simplifies the process of solving for D1 pooling equilibria.

\subsection{Uninformative equilibria}\label{sec4.5}

There exists an equilibrium where all sender types choose the uninformative experiment if the cost of experiments is sufficiently high. Specifically, if \(V(\mu_1,\mu_1)=0\), no sender type has a positive payoff guarantee. Therefore, there is an uninformative equilibrium where the receiver attributes any deviation to the lowest sender type. This equilibrium, however, may be ruled out by the D1 criterion using large deviations. The analysis is akin to that of pooling equilibria. In Lemma \ref{lmm.a9}, we show that there exists an experiment \(\pi'\in\Pi^\circ\) that gives all sender types zero payoffs, i.e., \(f(\pi',\mu_\theta)=0\) for all \(\theta\), and at \(\pi'\), the marginal rate of substitution is decreasing in the sender's type. Hence, we can find a deviation \(\tilde{\pi}\) such that \(f(\tilde{\pi},\mu_N)>0\), and \(f(\tilde{\pi},\mu_\theta)<0\) for all \(\theta<N\). If \(\tilde{\pi}\) is persuasive at belief \(\mu_N\), the uninformative equilibrium is ruled out by the D1 criterion. The following proposition shows that there exists an uninformative D1 equilibrium if and only if \(\mu_N\) is sufficiently small.

\begin{pro}\label{pro11}
	Let \(C_g>\hat{K}(C_b)\). There exists \(\bar{\mu}\in(0,\bar{\beta})\) such that an uninformative D1 equilibrium exists if and only if
	\begin{enumerate*}[label=(\roman*)]
		\item \(V(\mu_1,\mu_1)=0\), and
		\item \(\mu_N\leq\bar{\mu}\).
	\end{enumerate*}
	If \(V(\mu_N,\mu_N)=0\), the unique D1 equilibrium outcome is that all sender types choose the uninformative experiment.
\end{pro}

If \(\mu_N\leq\bar{\mu}\), all experiments \((\hat{\mathbf{p}}(q),q)\) are robust to large deviations. This is in line with Proposition \ref{pro5}, which states that, fixing \(\mu_N\), an experiment \((\hat{\mathbf{p}}(q),q)\) is robust to large deviations if and only if \(q\) is large. Moreover, when an uninformative equilibrium exists, \(V(\mu_1,\mu_1)=0\), so \((\hat{\mathbf{p}}(q),q)\) is a D1 pooling equilibrium outcome for all \(q\in(0,\bar{q}]\). That is, the uninformative D1 equilibrium coexists with a continuum of D1 pooling equilibria. Although the set of D1 pooling equilibrium outcomes is not closed, the uninformative equilibrium outcome is a limit point in the sense that, as \(q\) goes to zero, the players' equilibrium payoffs and the joint distribution of the state and the receiver's action converge to those in the uninformative equilibrium.

\subsection{Separating equilibrium under single-crossing}

We now study the case that the single-crossing property holds. The following proposition shows that, if the cost of experiments is low, the D1 criterion selects a unique separating equilibrium outcome; if the cost of experiments is high, a set of the lowest types of the sender choose the uninformative experiment, while all other sender types choose distinct experiments.

\begin{pro}\label{pro6}
	Let \(C_g\leq\hat{K}(C_b)\). There exists a unique D1 equilibrium outcome \(\{\pi_\theta\}_{\theta\in\Theta}\) such that:
	\begin{enumerate}[label=(\roman*)]
		\item \(\pi_\theta\) is the uninformative experiment if \(V(\mu_\theta,\mu_\theta)=0\).
		\item \(\pi_\theta\) is persuasive at belief \(\mu_\theta\) if \(V(\mu_\theta,\mu_\theta)>0\).
		\item The lowest sender type's equilibrium payoff \(v_1^\star\) equals \(V(\mu_1,\mu_1)\). Every sender type \(\theta\geq 2\) such that \(V(\mu_\theta,\mu_\theta)>0\) chooses the experiment \(\pi_\theta=(p_\theta,\mathbf{Q}(\mu_\theta)p_\theta)\), where \(p_\theta\) is the unique solution of
		\begin{equation}\label{eq.5}
			\begin{aligned}
				v^\star_\theta = &\max_{p} f((p,\mathbf{Q}(\mu_\theta)p),\mu_\theta) \\
				&s.t.~ f((p,\mathbf{Q}(\mu_\theta)p),\mu_{\theta-1})\leq v_{\theta-1}^\star.
			\end{aligned}
		\end{equation}
	\end{enumerate}
\end{pro}

The uniqueness result under single-crossing is reminiscent of \citeauthor{cho_sobel_1990}'s (\citeyear{cho_sobel_1990}) results despite differences in our settings.\footnote{\citet{cho_sobel_1990} assume that the signal space is the product of intervals, and higher sender types are more willing to send higher signals (with respect to the product order). In our game, the monotonicity assumption does not hold. Moreover, since the cost of any fully revealing experiment is infinite, the signal space consists of \(\Pi^\circ\) and the uninformative experiment, which is not compact. The structure of the unique D1 equilibrium outcome in our game also differs from that in \citet{cho_sobel_1990}. Recall that in \citet{cho_sobel_1990}, the equilibrium signal is nondecreasing in the sender's type, and some highest sender types may pool on the greatest signal.} Notice that \(V(\mu_\theta,\mu_\theta)\) is weakly increasing in the sender's type. Therefore, conditions (i) and (ii) imply that a set of the lowest sender types may choose the uninformative experiment. When the cost of experiments is so high that \(V(\mu_N,\mu_N)=0\), all sender types choose the uninformative experiment. Condition (iii) solves the equilibrium outcome by induction, starting from the lowest sender type. The intuition is that no sender type can increase her payoff without violating a lower sender type's incentive constraints. The single-crossing property allows us to solve a simplified maximization problem (\ref{eq.5}) which includes only the incentive constraint of the adjacent lower type.

Let us illustrate this in Figure \ref{fig.5}. Suppose that there are two sender types \(h>l\), and both of them have a positive payoff guarantee. That is, by choosing some experiment \(\hat{\pi}_\theta\), the type \(\theta\) sender is able to get payoff \(V(\mu_\theta,\mu_l)>0\) regardless of the receiver's belief. In a separating equilibrium, the receiver knows the sender's type at the interim stage. Hence, for the low-type sender, she chooses the experiment \(\hat{\pi}_l\) and receives her payoff guarantee in any separating equilibrium.

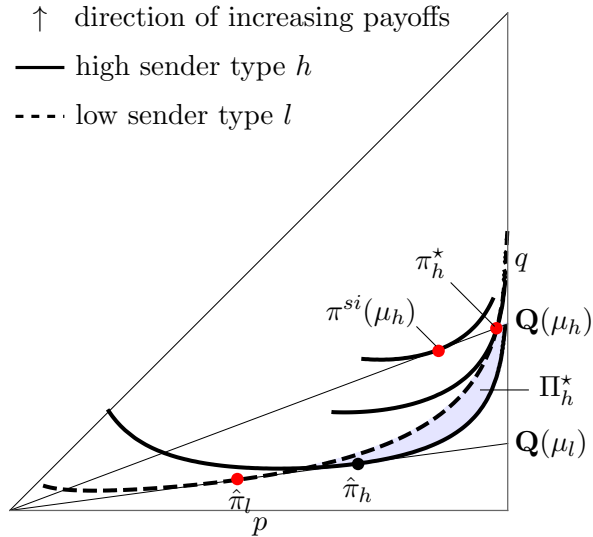
\begin{figure}[th]
	\centering
	\begin{tikzpicture}[scale=.65]
        \small
        \node[inner sep=0pt] at (5,5) {\includegraphics[width=.405\textwidth]{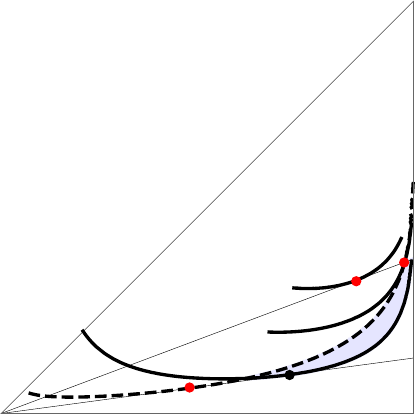}};
        \node[below] at (5,0) {$p$};
        \node[right] at (10,5) {$q$};
        \node[right] at (10,1.3) {$\mathbf{Q}(\mu_l)$};
        \node[right] at (10,3.8) {$\mathbf{Q}(\mu_h)$};
        \draw (9.8,3.7) -- (8.9,4.6) node[above left=-.1] {$\pi_h^\star$};
        \draw (8.6,3.3) -- (8.2,3.7) node[above left=-.2] {$\pi^{si}(\mu_h)$};
        \node[below] at (4.6,.59) {$\hat{\pi}_l$};
        \node[below] at (7,.82) {$\hat{\pi}_h$};
        \draw (9.5,2.4) -- (10.5,2.4) node[right] {$\Pi_h^\star$};
        \draw[very thick] (0,9) -- (1,9) node[right] {high sender type \(h\)};
        \draw[very thick,dashed] (0,8) -- (1,8) node[right] {low sender type \(l\)};
        \node at (.5,10) {\(\uparrow\)};
        \node[right] at (1,10) {direction of increasing payoffs};
    \end{tikzpicture}
	\caption{The unique D1 equilibrium outcome under single-crossing}\label{fig.5}
\end{figure}

For the high-type sender, her choice of experiment \(\pi_h=(p,q)\) has to satisfy three constraints. First, it must give her at least her payoff guarantee, i.e., \(f(\pi_h,\mu_h)\geq V(\mu_h,\mu_l)\). Second, it cannot be a profitable deviation for the low-type sender, i.e., \(f(\pi_h,\mu_l)\leq f(\hat{\pi}_l,\mu_l)\). Finally, it must be persuasive at the receiver's interim belief \(\mu_h\), i.e., \(\frac{q}{p}\leq\mathbf{Q}(\mu_h)\). The shaded area in Figure \ref{fig.5} shows the set of experiments \(\Pi_h^\star\) satisfying these constraints. Among them, \(\pi_h^\star\) uniquely maximizes the high-type sender's payoff. The solid curve through \(\pi_h^\star\) shows the high-type sender's indifference curve.

Any experiment \(\pi_h\in\Pi_h^\star\) is supported by a separating equilibrium, but the unique D1 equilibrium outcome is the one where \(\pi_h=\pi_h^\star\). Consider a separating equilibrium where \(\pi_h\neq\pi_h^\star\). We can find an experiment \(\tilde{\pi}\) slightly below the experiment \(\pi^\star_h\) which gives the high-type sender a higher payoff than \(\pi_h\). That is, \(\tilde{\pi}\) is a profitable deviation for the high type if the receiver's interim belief is high but is strictly equilibrium dominated for the low type. Hence, the D1 criterion requires that the receiver attributes this deviation to the high type. But given this off-path belief, \(\tilde{\pi}\) is a profitable deviation for the high type. Therefore, the separating equilibrium does not satisfy D1. In contrast, when the high-type sender is prescribed to choose \(\pi_h^\star\), any deviation that is profitable for the high-type sender is also profitable for the low-type sender. Hence, the D1 criterion is silent about the receiver's interim belief.

For the sender, the D1 equilibrium outcome is Pareto optimal among all separating equilibrium outcomes, that is, there does not exist a separating equilibrium outcome that gives all types of the sender weakly higher payoffs and at least one sender type a strictly higher payoff.\footnote{It is easy to see that the result holds for partial pooling equilibrium outcomes where types pool only on the uninformative experiment. That is, \(\pi_i\neq\pi_j\) for all \(i\neq j\) and \(\pi_i,\pi_j\in\Pi^\circ\).} This follows from Proposition \ref{pro6}, which states that it is not possible to increase any sender type's equilibrium payoff without also raising the equilibrium payoff of the lower adjacent type. But in all separating equilibria, the lowest sender type has an equilibrium payoff equal to her payoff guarantee.

\begin{crl}\label{crl8}
	Let \(C_g\leq\hat{K}(C_b)\). The D1 equilibrium outcome is Pareto optimal among all separating equilibrium outcomes.
\end{crl}

\subsection{Implications of the sender's private information}\label{sec5.4}

We compare the informativeness of the equilibrium experiment when the sender has private information to that when she has no private information. Let us revisit the costly persuasion problem studied by \citet{gentzkow_kamenica_2014}. The sender and the receiver have common prior \(\mu\), and there is no sender private information. This symmetric information game has a unique equilibrium in which the sender chooses some experiment, which we denote by \(\pi^{si}(\mu)\). If \(V(\mu,\mu)>0\), that is, if the average cost of experiments is low, \(\pi^{si}(\mu)\) is persuasive at belief \(\mu\); otherwise, \(\pi^{si}(\mu)\) is the uninformative experiment.

When the single-crossing property holds, every sender type bar the lowest chooses a Blackwell more informative experiment \(\pi^\star_\theta\) in the separating equilibrium than \(\pi^{si}(\mu_\theta)\) in order to prevent lower types from mimicking. For example, in Figure \ref{fig.5}, \(\pi^{si}(\mu_h)\) is the experiment the high-type sender would choose if the noisy signal were publicly observed. But with private information, the high-type sender chooses a Blackwell more informative experiment \(\pi^\star_h\) in the equilibrium. Therefore, the receiver learns more about the state of the world when the noisy signal is privately observed by the sender compared to when the noisy signal is publicly observed.

When the single-crossing property fails, the set of D1 pooling equilibrium outcomes is not Blackwell ordered. Recall that the set of pooling equilibrium outcomes is a continuum \(\{(\hat{\mathbf{p}}(q),q):q\in[\underline{q},\bar{q}]\}\), where the upper bound \(\bar{q}\) satisfies \(\frac{\hat{\mathbf{p}}(\bar{q})}{\bar{q}}=\mathbf{Q}(\mu_0)\). The pooling equilibrium outcome \((\hat{\mathbf{p}}(\bar{q}),\bar{q})\) is Blackwell less informative than \(\pi^{si}(\mu_0)\). That is, the receiver learns less in this equilibrium compared to a symmetric information benchmark where no one observes the noisy signal. Other pooling equilibrium outcomes are not Blackwell comparable to \(\pi^{si}(\mu_0)\).

The following proposition summarizes the results.

\begin{pro}\label{pro7}
	If \(C_g>\hat{K}(C_b)\) and the set of D1 pooling equilibrium outcomes \(\{(\hat{\mathbf{p}}(q),q):q\in[\underline{q},\bar{q}]\}\) is nonempty, \((\hat{\mathbf{p}}(\bar{q}),\bar{q})\) is Blackwell less informative than \(\pi^{si}(\mu_0)\). If \(C_g\leq\hat{K}(C_b)\), in the unique separating equilibrium, \(\pi_\theta^\star\) is Blackwell more informative than \(\pi^{si}(\mu_\theta)\) for all sender types \(\theta\) such that \(\pi_\theta^\star\neq\pi^{si}(\mu_\theta)\).
\end{pro}

The increase (decrease) in the experiment's informativeness does not translate to an increase (decrease) in the receiver's equilibrium payoff. In the separating equilibrium and the particular pooling equilibrium where the sender chooses \((\hat{\mathbf{p}}(\bar{q}),\bar{q})\), as well as in the symmetric information benchmark, the receiver's payoff is zero, which is the same payoff he gets from taking the low action without learning. This is a common feature in Bayesian persuasion with discrete receiver actions.\footnote{In \citet{Hedlund_2024}, the receiver takes a continuous action, and he strictly benefits from sender private information in a separating equilibrium.}

However, if the receiver can commit to a higher decision threshold than his sequentially rational threshold (hence he strictly prefers the high action when choosing it),\footnote{Since persuasion has otherwise no value to him, the receiver is willing to commit to a higher threshold if doing so is possible. This assumption of commitment is also natural in many applications where the decision threshold is set exogenously, or where the receiver is required to exercise caution when taking an action that differs from the default. For example, federal laws regulate evidence requirements for drug approval in the US (see, for example, \citet{darrow_dhruva_redberg_2021} for a review), and in a court of law, the accused is convicted only if the prosecution can prove that he/she is guilty beyond a reasonable doubt.} Proposition \ref{pro7} implies that 
\begin{enumerate*}[label=(\roman*)]
	\item if bad news is costlier, the receiver's equilibrium payoff in the separating equilibrium is weakly higher than that in the benchmark where the sender's signal is public;
	\item if good news is costlier, the receiver's equilibrium payoff is strictly lower in some pooling equilibria compared to that in the benchmark where the sender observes no noisy signal.
\end{enumerate*}
Moreover, the receiver's payoff in the public-signal benchmark is weakly higher than that in the no-signal benchmark. That is, providing public information never hurts the receiver. Therefore, either eliminating the sender's private information or making it public benefits the receiver when good news is costlier, but hurts the receiver when bad news is costlier.

\section{Discussions}\label{sec8}

\subsection{Partial pooling equilibria}\label{sec8.1}

When the single-crossing property fails, partial pooling equilibria may satisfy the D1 criterion. We present in this section some properties of these equilibria.

The nearby deviation argument still applies. Therefore, if an experiment \(\pi=(p,q)\) is chosen by two sender types \(h>l\), \(p=\hat{\mathbf{p}}(q)\). Moreover, any sender type \(m\in[l,h]\) must also choose \(\pi\). Suppose that \(\pi_m\neq\pi\). Since \(f(\pi,\mu)-f(\pi_m,\mu)\) is an affine function of \(\mu\), it must be the case that all sender types are indifferent between \(\pi\) and \(\pi_m\). That is, \(\pi_m\) is the experiment \(\pi'\) in Figure \ref{fig.6} where the indifference curves of all sender types through \(\pi\) intersect again. \(\pi_m\) is persuasive at belief \(\mu_m\), since it is chosen only by the type \(m\) sender. Therefore, there exists a deviation \(\tilde{\pi}\) that is persuasive at belief \(\mu_h\), gives the type \(h\) sender a higher payoff than \(\pi\) but is strictly equilibrium dominated for all sender types below \(h\). By the D1 criterion, the receiver should assign at least probability \(\mu_h\) to the high state after seeing this deviation, but this makes \(\tilde{\pi}\) a profitable deviation for the type \(h\) sender. Finally, as is shown in Lemma \ref{lmm.a8}, the sender's payoff from choosing \((\hat{\mathbf{p}}(q),q)\) increases in \(q\) for all types, hence in any partial pooling D1 equilibrium, there is a single interval of sender types \(\{\underline{\theta},\dots,\bar{\theta}\}\) who pool on the same experiment, while all other sender types choose distinct experiments.

If \(V(\mu_1,\mu_1)=0\), no experiment that is persuasive at belief \(\mu_1\) gives the lowest sender type a positive payoff, but any experiment \((\hat{\mathbf{p}}(q),q)\) gives her a positive payoff. Therefore, the lowest sender type is one of the pooling types. That is, any D1 partial pooling equilibrium features a set of the lowest sender types who choose the same experiment when the cost of experiments is high.

The large deviation argument also applies to partial pooling equilibria. That is, for all experiments \(\tilde{\pi}\) that are persuasive at belief \(\mu_{\bar{\theta}}\), \(f(\tilde{\pi},\mu_{\bar{\theta}})>f(\pi,\mu_{\bar{\theta}})\) implies \(f(\tilde{\pi},\mu_\theta)>f(\pi,\mu_\theta)\) for all \(\theta<\bar{\theta}\). By Proposition \ref{pro5}, this is equivalent to a lower bound \(\hat{q}_{\bar{\theta}}\) on \(q\), but this lower bound depends on the highest pooling type \(\bar{\theta}\).

\subsection{More precise sender private information}\label{sec5.2}

We relax the assumption that \(\mu_N<\bar{\beta}\) and assume instead that \(\mu_N>\bar{\beta}>\mu_0\). That is, the receiver's interim belief \(\beta\) can exceed his sequentially rational threshold \(\bar{\beta}\). Hence, it is possible that the receiver takes the high action regardless of the experiment's outcome. Specifically, if \(\mathbf{B}(\beta,\pi,b)\geq\bar{\beta}\), or equivalently, if \(\beta\geq\mathbf{Q}^{-1}\left(\frac{1-q}{1-p}\right)\), the receiver with interim belief \(\beta\) will take the high action even if he observes the bad outcome from the experiment \(\pi=(p,q)\). If this is the case, the sender's payoff from choosing \(\pi\) is \(1-c(\pi|\mu_\theta)\).

The new assumption does not affect the nearby deviation argument. We show in Lemma \ref{lmm.b8} that the payoff function \(1-c(\pi|\mu_\theta)\) satisfies the single-crossing property for all values of \(C_g\) and \(C_b\). Therefore, there is no D1 equilibrium in which two or more sender types choose an experiment \(\pi\), and the receiver always takes the high action regardless of the outcome of \(\pi\). Lemma \ref{lmm.b9} generalizes Lemma \ref{lmm4} to the case of \(\mu_N>\bar{\beta}\). Hence, when the single-crossing property holds, all sender types (except for those who choose the uninformative experiment) choose different experiments.

When the single-crossing property fails, no pooling or uninformative equilibrium satisfies the D1 criterion because of large deviations. Consider again the example in Figure \ref{fig.6}, where \(\pi\) is a pooling equilibrium outcome that is robust to nearby deviations, and \(\tilde{\pi}=(\tilde{p},\tilde{q})\) is a large deviations. We consider three possible cases of the receiver's interim belief \(\beta(\tilde{\pi})\).

\begin{enumerate}
	\item If \(\beta(\tilde{\pi})<\mathbf{Q}^{-1}\left(\frac{\tilde{q}}{\tilde{p}}\right)\), the sender gets \(-c(\tilde{\pi}|\mu_\theta)\) from choosing \(\tilde{\pi}\), so both sender types are better off with \(\pi\).
	\item If \(\mathbf{Q}^{-1}\left(\frac{\tilde{q}}{\tilde{p}}\right)\leq \beta(\tilde{\pi})<\mathbf{Q}^{-1}\left(\frac{1-\tilde{q}}{1-\tilde{p}}\right)\), the sender's payoff from choosing \(\tilde{\pi}\) is \(f(\pi,\mu_\theta)\), so deviating to \(\tilde{\pi}\) is profitable only for the high-type sender, and it gives the low-type sender a strictly lower payoff than \(\pi\).
	\item If \(\beta(\tilde{\pi})\geq\mathbf{Q}^{-1}\left(\frac{1-\tilde{q}}{1-\tilde{p}}\right)\), the sender's payoff from choosing \(\tilde{\pi}\) is \(1-c(\tilde{\pi}|\mu_\theta)>f(\tilde{\pi},\mu_\theta)\). Since the sender's payoff from choosing \(\pi\) is \(f(\pi|\mu_\theta)=f(\pi^\prime|\mu_\theta)\), deviating to \(\tilde{\pi}\) is profitable for both types of the sender if \(\tilde{\pi}\) is sufficiently close to \(\pi^\prime\).
\end{enumerate}

By assumption, \(\mathbf{Q}(\mu_h)>1>\frac{\tilde{q}}{\tilde{p}}\). Hence, the second case is true for a nonempty subset of the receiver's interim beliefs. Therefore, the high-type sender is keener to deviate to \(\tilde{\pi}\) than the low-type sender, and the D1 criterion requires that \(\beta(\tilde{\pi})=\mu_h\). However, this falls into either the second or the third case where \(\tilde{\pi}\) is a profitable deviation for the high-type sender. Hence, the pooling equilibrium does not satisfy the D1 criterion.

Partial pooling D1 equilibria may continue to exist. As we show earlier, an interval of sender types \(\{\underline{\theta},\dots,\bar{\theta}\}\) pool on some experiment \((\hat{\mathbf{p}}(q),q)\), while all other sender types choose distinct experiments. By the large deviation argument above, \(\mu_{\bar{\theta}}<\bar{\beta}\). That is, there is an upper bound on the highest pooling type.

\subsection{Limiting cases}\label{sec8.3}

In the classic signaling model, there is a discontinuity as the sender's private information becomes degenerate or the differential cost of sending signals diminishes. These discontinuities persist in our game, where the signal is a public experiment that reveals information about the payoff relevant state of the world.

When there is no sender private information, the symmetric information game with common prior \(\mu_h\) has a unique equilibrium outcome \(\pi^{si}(\mu_h)\). Now suppose that the sender has slight private information such that her prior is \(\mu_l<\mu_h\) with an infinitesimal probability. One of the following happens depending on the cost of experiments. If \(C_g\leq\hat{K}(C_b)\), the game has a unique separating equilibrium, where the high-type sender chooses a more informative experiment than \(\pi^{si}(\mu_h)\). The separating outcome is independent of the probability of the low type sender, hence it does not converge to \(\pi^{si}(\mu_h)\) as the probability goes to zero. This is akin to the discontinuity observed in the classic signaling model. If \(C_g>\hat{K}(C_b)\), the game has a continuum of pooling equilibria, and all pooling equilibrium outcomes satisfy \(p=\hat{\mathbf{p}}(q)\) and are therefore separated from \(\pi^{si}(\mu_h)\). Therefore, although the failure of single-crossing leads to multiplicity of equilibrium outcomes, none of them converges to the equilibrium outcome of the symmetric information game.

When experiments are costless, all equilibria are persuasive and pooling, and the set of D1 pooling equilibrium outcomes is \(\{(1,q):q\in[\mathbf{Q}(\mu_1),\mathbf{Q}(\mu_0)]\}\).\footnote{\citet{hedlund_2017} assumes that the receiver's action is continuous and shows that D1 equilibria are either separating or fully revealing. In our model, the receiver takes a binary action. When experiments are costless, all D1 equilibria are pooling, and only the bad state is revealed. This is consistent with the findings of \citet{Kosenko_2023}.} That is, the bad outcome reveals the bad state. Now consider a sequence of costly persuasion games with diminishing costs such that \(\frac{C_g^n}{C_b^n}\) is uniformly bounded, where \(C_g^n\) and \(C_b^n\) denote the costs of good news and bad news in the \(n\)'th game, respectively. Recall from Proposition \ref{pro3} that \(\lim_{C_b\downarrow 0}\hat{K}'(C_b)=\infty\). Therefore, as the costs diminish, the game has a unique separating equilibrium outcome. In the limit, it converges to a strategy profile where each sender type \(\theta\) chooses a distinct experiment \((1,\mathbf{Q}(\mu_\theta))\).

\subsection{Microfoundation of the log-likelihood ratio cost}\label{sec7}

The log-likelihood ratio cost of experiments can be microfounded by a Wald sequential sampling problem. The sender acquires noisy, binary signals about the state of the world before making an irreversible decision to stop. An experiment in our game is equivalent to a threshold stopping rule in the sampling problem, that is, the sender stops acquiring signals at the first instance her belief reaches some thresholds. Equating the cost of an experiment with the expected cost of acquiring signals in the sampling problem yields the log-likelihood ratio costs.

Consider the following stopping problem in discrete time.\footnote{A continuous time version of the model is a special case of that in \citet{morris_strack_2019}. The sender observes a Brownian motion with a state-dependent drift term, and her flow cost \(c(p_t)\) is a function of her instantaneous belief \(p_t\) on the good state. Motivated by our examples, it is straightforward to set \(c(p_t)=p_tc_g+(1-p_t)c_b\). That is, conditional on the state \(\omega\), the flow cost of experimenting is a constant \(c_\omega\). Then the log-likelihood ratio cost function is derived using Theorem 1 of \citet{morris_strack_2019}.} There is a binary state of the world \(\omega\in\{G,B\}\), and a sequence of binary signals \((s_n)_{n=1}^\infty\) such that each \(s_n\in\{g,b\}\). Let \(\mu_0\) be the probability of the good state, and conditional on the state, signals are distributed iid such that \(\mathbb{P}(s_n=g|\omega=G)=\mathbb{P}(s_n=b|\omega=B)=\alpha>\frac{1}{2}\). Both the state of the world and the signals are realized at the outset of the game and are not observed by the sender. 

We model the sender's information acquisition as follows. At a history \(h_n=(s_1,s_2,\dots,s_n)\), that is, the sender has acquired signals \(s_1,s_2,\dots,s_n\) and has not yet stopped, the sender chooses between acquiring an additional signal \(s_{n+1}\) and irreversibly stopping signal acquisition. The cost of acquiring the signal \(s_{n+1}\) is a random variable that equals \(c_g\) if \(s_{n+1}=g\) and \(c_b\) if \(s_{n+1}=b\). Hence, the total cost of signals if the sender stops at history \(h_n\) is \(c_g n_g(h_n)+c_b n_b(h_n)\), where \(n_g(h_n)\) and \(n_b(h_n)\) denote the number of good and bad signals in \(h_n\), respectively.

A strategy of the sender is a stopping time adapted to the natural filtration generated by histories. Specifically, consider the following \emph{threshold strategy} \(\tau\) of the sender: she stops at the first history where the difference between the number of good and bad signals equals some threshold values \(\underline{n}<0\) or \(\bar{n}>0\). Notice that the sender stops in finite time with probability one. Hence, interpreting the event \(n_g(h_\tau)-n_b(h_\tau)=\bar{n}\) as the good outcome and the event \(n_g(h_\tau)-n_b(h_\tau)=\underline{n}\) as the bad outcome, the strategy \(\tau\) is equivalent to an experiment \(\pi=(p,q)\) in the our model of costly persuasion, where \(p\) and \(q\) are the probabilities of the good outcome conditional on the good and bad state, respectively.

The sender's posterior belief on the good state induced by the stopping rule \(\tau\) is a random variable \(\hat{\mu}\), which equals \(\mu_{\bar{n}}:=\left[1+\frac{1-\mu_0}{\mu_0}\left(\frac{1-\alpha}{\alpha}\right)^{\bar{n}}\right]^{-1}\) following the good outcome, or \(\mu_{\underline{n}}:=\left[1+\frac{1-\mu_0}{\mu_0}\left(\frac{1-\alpha}{\alpha}\right)^{\underline{n}}\right]^{-1}\) following the bad outcome. This is the same as the posterior belief induced by the experiment \(\pi\).

\begin{pro}\label{pro9}
	The expected cost of implementing the strategy \(\tau\), \(\mathbb{E}[c_g n_g(h_\tau)+c_b n_b(h_\tau)],\) equals \(\mathbb{E}[H(\mu_0)-H(\hat{\mu})]\), where
	\begin{equation*}
		H(\mu)=\frac{1}{2\alpha-1}\left[\ln\left(\frac{\alpha}{1-\alpha}\right)\right]^{-1}\left[\bar{c}_g\mu\ln\left(\frac{1-\mu}{\mu}\right)+\bar{c_b}(1-\mu)\ln\left(\frac{\mu}{1-\mu}\right)\right],
	\end{equation*}
	and \(\bar{c}_g=\alpha c_g+(1-\alpha)c_b\), \(\bar{c}_b=(1-\alpha)c_g+\alpha c_b\).
\end{pro}

Hence, letting \(C_g=\frac{1}{2\alpha-1}\left[\ln\left(\frac{\alpha}{1-\alpha}\right)\right]^{-1}\bar{c}_g\) and \(C_b=\frac{1}{2\alpha-1}\left[\ln\left(\frac{\alpha}{1-\alpha}\right)\right]^{-1}\bar{c_b}\), the expected cost of implementing the strategy \(\tau\) equals the log-likelihood ratio cost of the experiment \(\pi\).

\subsection{Shannon entropy cost of experiments}\label{sec8.2}

Our result that the single-crossing property may fail in a costly Bayesian persuasion game with a partially informed sender is not specific to the log-likelihood ratio cost of experiments we use. The following proposition shows a similar result in a game with the Shannon entropy cost of experiments. That is, \(c(\pi|\mu)=\mathbb{E}[H(\mu)-H(\hat{\mu})]\), where \(H(\mu)=-C\left[\mu\ln\mu+(1-\mu)\ln(1-\mu)\right]\), and \(C>0\).

\begin{pro}\label{pro10}
	The single-crossing property fails if experiments have the Shannon entropy cost.
	\begin{enumerate}[label=(\roman*)]
		\item There exists \(\tilde{\mathbf{p}}:(0,1)\to(0,1)\) such that \(\tilde{\mathbf{p}}(q)>q\), and \(MRS(\pi|\mu)\) is decreasing in \(\mu\) if \(p\leq\tilde{\mathbf{p}}(q)\), and \(MRS(\pi|\mu)\) first increases and then decreases in \(\mu\) if \(p>\tilde{\mathbf{p}}(q)\);
		\item For any sender types \(i<j\), there exists \(\tilde{\mathbf{p}}_{i,j}:(0,1)\to(0,1)\) such that \(\tilde{\mathbf{p}}_{i,j}(q)>\tilde{\mathbf{p}}(q)\), and \(MRS(\pi|\mu_i)>MRS(\pi|\mu_j)\) if \(p<\tilde{\mathbf{p}}_{i,j}(q)\), \(MRS(\pi|\mu_i)=MRS(\pi|\mu_j)\) if \(p=\tilde{\mathbf{p}}_{i,j}(q)\), and \(MRS(\pi|\mu_i)<MRS(\pi|\mu_j)\) if \(p>\tilde{\mathbf{p}}_{i,j}(q)\).
	\end{enumerate}
\end{pro}

Part (i) of Proposition \ref{pro10} shows that the single-crossing property fails at experiments that are very informative (i.e., \(p>\tilde{\mathbf{p}}(q)\)). In fact, the double-crossing property \citep{chen_ishida_suen_2022} is satisfied. As part (ii) shows, given any two sender types, there exists a locus of experiments where the indifference curves of the two sender types are tangent, and the high-type sender's indifference curve is more convex. Therefore, the D1 criterion may select an equilibrium that features pairwise pooling where each experiment is chosen by up to two sender types.

\section{Conclusions}

We study Bayesian persuasion with costly experiments and a partially informed sender. The receiver can learn about the state of the world from the sender's choice of experiment as well as from the outcome of the experiment. We show that this is a class of signaling games where the single-crossing property may fail, and the equilibrium outcome crucially depends on the cost of experiments.

We focus on the log-likelihood ratio cost where the costs of drawing good news and bad news can differ. If good news is not too costly compared to bad news, there exists a unique separating equilibrium outcome. The receiver fully learns the sender's private signal from her choice of experiment, and the sender chooses a Blackwell more informative experiment than what she would choose if the noisy signal were public. On the other hand, if good news is much costlier than bad news, the single-crossing property fails. The sender with more preferable private information may not want to separate by providing more information. We characterize pooling equilibrium outcomes and show that the receiver learns less about the state in some pooling equilibria than in the benchmark with no noisy signal.

These results have practical implications for different real-life applications. In drug approval, information asymmetry at the preclinical stage incentivizes pharmaceutical companies to run more informative clinical trials. Hence, it may be beneficial to public welfare. In startup funding, however, an entrepreneur's private knowledge about a proprietary technology can hurt investors. Therefore, it is important to solicit opinions from external experts.

\clearpage
\appendix

\section{Auxiliary Results and Proofs}

\subsection{Proof of Proposition \ref{pro3}}

\begin{proof}
	Let \(\mu\in(0,1)\) and \(\pi=(p,q)\in\Pi^\circ\). The cost of running the experiment \(\pi\) for the sender type who has prior belief \(\mu\) is
	\begin{align*}
		c(\pi|\mu) &= C_g\left[\mu\ln\frac{1-\mu}{\mu}-\mu p\ln\frac{(1-\mu)q}{\mu p}-\mu(1-p)\ln\frac{(1-\mu)(1-q)}{\mu(1-p)}\right] \\
		&\quad +C_b\left[(1-\mu)\ln\frac{\mu}{1-\mu}-(1-\mu)q\ln\frac{\mu p}{(1-\mu)q}-(1-\mu)(1-q)\ln\frac{\mu(1-p)}{(1-\mu)(1-q)}\right] \\
		&= C_g\mu\left[-p\ln\frac{q}{p}-(1-p)\ln\frac{1-q}{1-p}\right] + C_b(1-\mu)\left[-q\ln\frac{p}{q}-(1-q)\ln\frac{1-p}{1-q}\right] \\
		&=C_g\mu\left[p\ln\frac{p}{q}+(1-p)\ln\frac{1-p}{1-q}\right] + C_b(1-\mu)\left[q\ln\frac{q}{p}+(1-q)\ln\frac{1-q}{1-p}\right] \\
		&= C_g\mu D_{KL}(P||Q) + C_b(1-\mu)D_{KL}(Q||P),
	\end{align*}
	where \(D_{KL}\) is the Kullback-Leibler divergence \citep{kullback_leibler_1951}, and \(P\) and \(Q\) are Bernoulli distributions with success rates \(p\) and \(q\), respectively. Notice that it is an affine function of \(\mu\), and it is increasing in \(\mu\) if \(\frac{C_g}{C_b}>\frac{D_{KL}(P||Q)}{D_{KL}(Q||P)}\) and decreasing in \(\mu\) if \(\frac{C_g}{C_b}<\frac{D_{KL}(P||Q)}{D_{KL}(Q||P)}\).

	We now compute the marginal rate of substitution \(MRS(\pi|\mu)\). Observe that
	\begin{gather}
		\frac{\partial f(\pi,\mu)}{\partial p} = - \underbrace{C_b\frac{p-q}{p(1-p)}}_{=:A_1(p,q)} - \mu \underbrace{\left[-1+C_g\ln\frac{p(1-q)}{(1-p)q}-C_b\frac{p-q}{p(1-p)}\right]}_{=:A_2(p,q)}, \label{eq.f1}\\
		\frac{\partial f(\pi,\mu)}{\partial q} = \underbrace{1+C_b\ln\frac{p(1-q)}{(1-p)q}}_{=:A_3(p,q)} + \mu \underbrace{\left[-1+C_g\frac{p-q}{q(1-q)}-C_b\ln\frac{p(1-q)}{(1-p)q}\right]}_{=:A_4(p,q)}. \label{eq.f2}
	\end{gather}
	Hence,
	\begin{equation*}
		MRS(\pi|\mu) = -\frac{\partial f(\pi,\mu)/\partial p}{\partial f(\pi,\mu)/\partial q} = \frac{A_1(p,q)+A_2(p,q)\mu}{A_3(p,q)+A_4(p,q)\mu}.
	\end{equation*}
	Taking derivative with respect to \(\mu\),
	\begin{equation*}
		\frac{\partial}{\partial\mu} MRS(\pi|\mu) = \frac{\Delta(p,q)}{\big[A_3(p,q)+A_4(p,q)\mu\big]^2},
	\end{equation*}
	where
	\begin{align*}
		\Delta(p,q) &= A_2(p,q)A_3(p,q)-A_1(p,q)A_4(p,q) \\
		&= (C_g-C_b)\ln\frac{p(1-q)}{(1-p)q} + C_gC_b\left[\left(\ln\frac{p(1-q)}{(1-p)q}\right)^2-\frac{(p-q)^2}{pq(1-p)(1-q)}\right] - 1.
	\end{align*}
	Notice that \(\Delta(p,q)\) is independent of \(\mu\). Therefore, the marginal rate of substitution \(MRS(\pi|\mu)\) is monotonic in the sender's prior belief \(\mu\), and whether it is increasing or decreasing is determined by the sign of \(\Delta(p,q)\).

	Define a change of variable
	\begin{equation}\label{eq.a17}
		t=\frac{(1-p)q}{p(1-q)}.
	\end{equation}
	\(t\in(0,1)\) for all \((p,q)\in\Pi^\circ\). We can rewrite \(\Delta\) as a function of just \(t\). To simplify notation, we use the same letter \(\Delta\) to denote this function
	\begin{equation}\label{eq.a14}
		\Delta(t)=-(C_g-C_b)\ln t+C_gC_b\left[(\ln t)^2-\frac{(1-t)^2}{t}\right]-1.
	\end{equation}
	Taking derivative of (\ref{eq.a14}),
	\begin{equation}\label{eq.a15}
		\Delta'(t) = -\frac{C_g-C_b}{t} + \frac{C_gC_b}{t}\left(2\ln t-t+\frac{1}{t}\right).
	\end{equation}
	Hence, \(\Delta'(t)=0\) if and only if
	\begin{equation}\label{eq.a19}
		2\ln t-t+\frac{1}{t} = \frac{1}{C_b}-\frac{1}{C_g}.
	\end{equation}
	
	The left-hand side of (\ref{eq.a19}) is positive for all \(t\in(0,1)\) and decreasing in \(t\). As \(t\to 0\), it goes to infinity, and at \(t=1\), it equals zero. Therefore, if \(C_g\leq C_b\), \(\Delta'(t)>0\) for all \(t\in(0,1)\), so \(\Delta(t)<\Delta(1)=-1\) for all \(t\in(0,1)\). That is, the marginal rate of substitution at any experiment is decreasing in the sender's prior belief, and the single-crossing property holds.

	If \(C_g>C_b\), (\ref{eq.a19}) has a unique solution \(t^\star\in(0,1)\). \(\Delta(t)\) is single-peaked and obtains its maximum at \(t=t^\star\). Therefore, the single-crossing property holds if and only if \(\Delta(t^\star)\leq 0\). Notice that \(\Delta(1)=-1\) and \(\Delta'(1)<0\), so \(\Delta(t^\star)>-1\). By envelope theorem,
	\begin{equation}\label{eq.d9}
		\frac{\partial\Delta(t^\star)}{\partial C_g} = -\ln t^\star + C_b\left[(\ln t^\star)^2-\frac{(t^\star-1)^2}{t^\star}\right] \geq \frac{\Delta(t^\star)+1}{C_g} > 0.
	\end{equation}
	Notice that \(y'(x)=\frac{y(x)+1}{x}\) solves a linear function \(y\) of \(x\). Therefore, (\ref{eq.d9}) implies that \(\Delta(t^\star)\) increases in \(C_g\) at least as fast as a linear function. Moreover, given any \(C_b>0\), as \(C_g\downarrow C_b\), \(t^\star\to 1\), and \(\Delta(t^\star)<0\). Therefore, \(\Delta(t^\star)=0\) has a unique solution \(C_g=\hat{K}(C_b)\), and the single-crossing property holds if and only if \(C_g\leq\hat{K}(C_b)\).

	We claim that
	\begin{equation}\label{eq.b3}
		\hat{K}(C_b)=\left(\frac{\hat{x}(C_b)^2}{1+\hat{x}(C_b)}-\frac{1}{C_b}\right)^{-1},
	\end{equation}
	where \(x=\hat{x}(C_b)>0\) solves
	\begin{equation} \label{eq.b4}
		x-\ln(1+x)=\frac{1}{C_b}.
	\end{equation}
	The left-hand side of (\ref{eq.b4}) is increasing in \(x\) on \(\mathbb{R}_+\). It equals zero at \(x=0\) and goes to infinity as \(x\to\infty\). Hence, \(\hat{x}(C_b)\) and the right-hand side of (\ref{eq.b3}) are well defined.

	We verify the claim by showing that \(\Delta(t^\star)=0\) if \(C_g\) equals (\ref{eq.b3}). In (\ref{eq.a19}), substituting \(C_g\) using (\ref{eq.b3}), we have
	\begin{equation}\label{eq.d11}
		2\ln t-t+\frac{1}{t}=\frac{2}{C_b}-\frac{\hat{x}(C_b)^2}{1+\hat{x}(C_b)}.
	\end{equation}
	By (\ref{eq.b4}), it is easy to see that \(t^\star=\frac{1}{1+\hat{x}(C_b)}\) solves (\ref{eq.d11}). To simplify notation, we write \(x\) for \(\hat{x}(C_b)\). Then evaluating (\ref{eq.a14}) at \(t^\star\) yields
	\begin{align*}
		\Delta(t^\star) &= -\frac{C_b^2x^2-2C_b(1+x)}{C_bx^2-x-1}\ln(1+x) + \frac{C_b^2(1+x)}{C_bx^2-x-1}\left[(\ln(1+x))^2-\frac{x^2}{1+x}\right]-1 \\
		&= -\frac{C_b^2x^2-2C_b(1+x)}{C_bx^2-x-1}\left(x-\frac{1}{C_b}\right) + \frac{C_b^2(1+x)}{C_bx^2-x-1}\left(x-\frac{1}{C_b}\right)^2 - \frac{C_b^2x^2}{C_bx^2-x-1}-1 \\
		&= \left(x-\frac{1}{C_b}\right)\frac{C_b(1+x)+C_b^2x}{C_bx^2-x-1}-\frac{C_b^2x^2}{C_bx^2-x-1}-1 \\
		&= \frac{C_bx^2-x-1}{C_bx^2-x-1}-1 = 0,
	\end{align*}
	which is the desired result.

	By (\ref{eq.b4}), \(\hat{x}\) is a smooth function of \(C_b\). Hence, by (\ref{eq.b3}), \(\hat{K}\) is a smooth function of \(C_b\). We are left to show that \(\hat{K}'(C_b)>0\), \(\hat{K}''(C_b)<0\), and \(\lim_{C_b\downarrow 0}\hat{K}'(C_b)=\infty\). Taking derivative of (\ref{eq.b4}),
	\begin{equation*}
		\hat{x}'(C_b)=-\frac{1}{C_b^2}\left(1-\frac{1}{1+\hat{x}(C_b)}\right)^{-1}.
	\end{equation*}
	Substituting this in the derivative of (\ref{eq.b3}), we have
	\begin{equation*}
		\hat{K}'(C_b) = \left(\frac{\hat{K}(C_b)}{C_b}\right)^2\frac{1}{1+\hat{x}(C_b)} > 0.
	\end{equation*}
	Applying l'H\^{o}pital's rule, we have
	\begin{equation*}
		\lim_{C_b\to 0}\frac{1}{\hat{K}'(C_b)}=\lim_{C_b\to 0}\frac{1}{1+\hat{x}(C_b)}=0.
	\end{equation*}
	That is, \(\lim_{C_b\downarrow 0}\hat{K}'(C_b)=\infty\). To calculate the second order derivative, we use the fact that \([(1+\hat{x}(C_b))\hat{K}'(C_b)]'=\hat{x}'(C_b)\hat{K}'(C_b)+(1+\hat{x}(C_b))\hat{K}''(C_b)\). Hence,
	\begin{align}
		\hat{K}''(C_b) &= \frac{[(1+\hat{x}(C_b))\hat{K}'(C_b)]'-\hat{x}'(C_b)\hat{K}'(C_b)}{1+\hat{x}(C_b)} \nonumber\\
		&= \frac{\hat{K}(C_b)^2}{(1+\hat{x}(C_b))C_b^3}\left[2\left(\frac{\hat{K}(C_b)}{C_b(1+\hat{x}(C_b)){}}-1\right)+\frac{1}{C_b\hat{x}(C_b)}\right].\label{eq.b10}
	\end{align}
	Substituting \(\hat{K}(C_b)\) using (\ref{eq.b3}) and \(\frac{1}{C_b}\) using (\ref{eq.b4}), and to shorten notation, writing \(x\) for \(\hat{x}(C_b)\), the bracket in (\ref{eq.b10}) becomes
	\begin{equation}\label{eq.b11}
		\frac{3x^2-(x^2+2x)\ln(1+x)-(1+x)(\ln(1+x))^2}{(1+x)\ln(1+x)-x}.
	\end{equation}
	The denominator of (\ref{eq.b11}) is positive for all \(x>0\). The third order derivative of the numerator of (\ref{eq.b11}) with respect to \(x\) is
	\begin{equation*}
		\frac{-2x(2+x)+2(1+x)\ln(1+x)}{(1+x)^3} < \frac{2[\ln(1+x)-x]}{(1+x)^2},
	\end{equation*}
	which is negative for all \(x>0\). Notice that evaluated at \(x=0\), the numerator of (\ref{eq.b11}) and its derivatives up to the third order are zero. Therefore, the numerator of (\ref{eq.b11}) is negative for all \(x>0\). Hence, (\ref{eq.b11}) is negative, and \(\hat{K}''(C_b)<0\).
\end{proof}

\subsection{Proof of Proposition \ref{pro3_part2}}

\begin{proof}
	If \(C_g>\hat{K}(C_b)\), \(\Delta(t^\star)>0\). Therefore, \(\Delta(t)\) has two zeros \(\hat{t}\) and \(\check{t}\) such that \(1>\hat{t}>t^\star>\check{t}>0\), and
	\begin{equation*}
		\Delta(t) \left\{\begin{array}{ll}
			<0 & \text{if \(t\in(0,\check{t})\cup(\hat{t},1)\)} \\
			=0 & \text{if \(t\in\{\check{t},\hat{t}\}\)} \\ 
			>0 & \text{if \(t\in(\check{t},\hat{t})\)}
		\end{array}\right..
	\end{equation*}
	Letting
	\begin{equation}\label{eq.b5}
		\hat{\mathbf{p}}(q) = \frac{q}{q+\hat{t}(1-q)},\qquad \check{\mathbf{p}}(q) = \frac{q}{q+\check{t}(1-q)}
	\end{equation}
	yields the desired result (\ref{eq.4}).

	Notice that
	\begin{equation*}
		\frac{\partial\Delta(t)}{\partial C_b} = \ln t + C_g\left[(\ln t)^2-\frac{(1-t)^2}{t}\right] < 0
	\end{equation*}
	for all \(t\in(0,1)\). Therefore, \(\hat{t}\) is decreasing in \(C_b\), and \(\check{t}\) is increasing in \(C_b\). Equivalently, \(\hat{\mathbf{p}}(q)\) is increasing in \(C_b\), and \(\check{\mathbf{p}}(q)\) is decreasing in \(C_b\) for all \(q\in(0,1)\). On the other hand,
	\begin{equation*}
		\frac{\partial\Delta(t)}{\partial C_g} = -\ln t + C_b\left[(\ln t)^2-\frac{(1-t)^2}{t}\right] = \frac{\Delta(t)+1-C_b\ln t}{C_g}.
	\end{equation*}
	Hence, at \(\hat{t}\) and \(\check{t}\),
	\begin{equation*}
		\frac{\partial\Delta(t)}{\partial C_g} = \frac{1-C_b\ln t}{C_g} > 0.
	\end{equation*}
	Therefore, \(\hat{t}\) is increasing in \(C_g\), and \(\check{t}\) is decreasing in \(C_b\). Equivalently, \(\hat{\mathbf{p}}(q)\) is decreasing in \(C_g\), and \(\check{\mathbf{p}}(q)\) is increasing in \(C_g\).
\end{proof}

\subsection{Proof of Lemma \ref{lmm4}}

\begin{proof}
	To obtain a contradiction, suppose that there exists a D1 equilibrium where both sender types \(i\) and \(j\) choose \(\pi=(p,q)\). Without loss of generality, let \(i\) be the highest sender type that chooses \(\pi\), hence \(\beta(\pi)<\mu_i\). For the rest of the proof, we assume that \(MRS(\pi|\mu_i)< MRS(\pi|\mu_j)\). A similar proof works if \(MRS(\pi|\mu_i)> MRS(\pi|\mu_j)\).

	Since \(\pi\) is persuasive at \(\beta(\pi)<\mu_i\), by continuity, there is a neighborhood of \(\pi\) wherein all experiments are persuasive at belief \(\mu_i\).

	For each sender type \(\theta\), denote by \(\mathfrak{q}_\theta\) her indifference curve through \(\pi\). That is, 
	\begin{equation*}
		f((p',\mathfrak{q}_\theta(p')),\mu_\theta)=f(\pi,\mu_\theta)
	\end{equation*}
	for all \(p'\). \(\mathfrak{q}_\theta\) is well defined on a small neighborhood of \(p\), and it is the solution to the initial value problem
	\begin{equation}\label{eq.d13}
		\mathfrak{q}_\theta'(p') = MRS((p',\mathfrak{q}_\theta(p'))|\mu_\theta)~\text{and}~\mathfrak{q}_\theta(p)=q.
	\end{equation}
	By Propositions \ref{pro3}, \(MRS(\pi'|\mu)\) is monotone in \(\mu\) for all \(\pi'\). Hence, there exists \(\varepsilon>0\) such that \(\mathfrak{q}_\theta(p')\) is strictly decreasing in \(\theta\) for all \(p'\in(p,p+\varepsilon)\).

	Hence, we can find \(\tilde{\pi}=(\tilde{p},\tilde{q})\) such that \(\tilde{\pi}\) is persuasive at belief \(\mu_i\), \(\tilde{p}\in(p,p+\varepsilon)\), and \(\tilde{q}\in(\mathfrak{q}_i(\tilde{p}),\mathfrak{q}_{i-1}(\tilde{p}))\). By construction, \(f(\tilde{\pi},\mu_i)>f((\tilde{p},\mathfrak{q}_i(\tilde{p})),\mu_i)=f(\pi,\mu_i)=v^\star_i\), and \(f(\tilde{\pi},\mu_k)<f((\tilde{p},\mathfrak{q}_k(\tilde{p})),\mu_k)=f(\pi,\mu_k)= v^\star_k\) for all \(k<i\). That is, \(\tilde{\pi}\) is a profitable deviation for the sender type \(i\) if \(\beta(\tilde{\pi})\geq\mu_i\), so \(\beta(\tilde{\pi})<\mu_i\). However, since \(\tilde{\pi}\) is strictly equilibrium dominated for all sender types \(k<i\), the D1 criterion requires that \(\beta(\tilde{\pi})\geq\mu_i\), a contradiction.
\end{proof}

\subsection{Results relating to the symmetric information benchmark}

We present in this section results relating to the symmetric information benchmark where the sender and the receiver have heterogeneous priors \(\mu\) and \(\beta\), respectively, and they agree to disagree. Lemma \ref{lmm.a1} shows that \(V(\mu,\beta)>0\) if and only if the ``average'' cost of experiments is sufficiently low. Two auxiliary functions, \(M(p,q)\) and \(N(p,q)\) are defined. Lemma \ref{lmm.a7} proves some properties of these functions that are necessary for the proof of Lemma \ref{lmm.a6}. Lemma \ref{lmm.a5} characterizes the equilibrium outcome of the symmetric information benchmark, which is used to prove Proposition \ref{pro7}.

\begin{lmm}\label{lmm.a1}
	In the symmetric information benchmark, \(V(\mu,\beta)>0\) if and only if \(\mathbf{F}(C_g,C_b,\mu,\beta)>0\), where
	\begin{align*}
		\mathbf{F}(C_g,C_b,\mu,\beta)=\mu+(1-\mu)\mathbf{Q}(\beta)+&\mu C_g\left[\ln \mathbf{Q}(\beta)+1-\mathbf{Q}(\beta)\right] \\
		&\quad-(1-\mu)C_b\left[\mathbf{Q}(\beta)\ln \mathbf{Q}(\beta)+1-\mathbf{Q}(\beta)\right].
	\end{align*}
	If \(\mathbf{F}(C_g,C_b,\mu,\beta)>0\), there exists a unique equilibrium outcome which is persuasive at belief \(\beta\). If \(\mathbf{F}(C_g,C_b,\mu,\beta)\leq 0\), the unique equilibrium outcome is the uninformative experiment.
\end{lmm}

\begin{proof}
	Notice that the sender's expected payoff \(f((p,q),\mu)\) is increasing in \(q\), and that \(\lim_{p\downarrow 0}f((p,\mathbf{Q}(\beta)p),\mu)=0\). Therefore, \(V(\mu,\beta)\) is the value of
	\begin{equation}\label{eq.2}
		\max_{p\in[0,1]}f((p,\mathbf{Q}(\beta)p),\mu).
	\end{equation}
	The objective function of (\ref{eq.2}) is continuous and concave. Therefore, it is solved by the first order condition. To simplify notation, we will denote briefly the derivative \(\frac{d}{dp}f((p,\mathbf{Q}(\beta)p),\mu)\) by \(f'(p)\) and \(\mathbf{Q}(\beta)\) by \(Q\). Observe that that
	\begin{equation}\label{eq.b2}
		f'(p) = \mu M(p,Qp)+(1-\mu)N(p,Qp),
	\end{equation}
	where
	\begin{gather*}
		M(p,q) := 1+C_g\left[\frac{p-q}{p(1-q)}-\ln\frac{p(1-q)}{(1-p)q}\right], \\
		N(p,q) := \frac{q}{p}\left(1-C_b\left[\frac{p-q}{(1-p)q}-\ln\frac{p(1-q)}{(1-p)q}\right]\right).
	\end{gather*}
	Since \(\lim_{p=1}f'(p)=-\infty\), the first order condition \(f'(p)=0\) has an interior solution \(\hat{p}\in(0,1)\) if and only if
	\begin{equation*}
		\lim_{p\downarrow 0} f'(p) = \mathbf{F}(C_g,C_b,\mu,\beta) > 0.
	\end{equation*}
	If \(\mathbf{F}(C_g,C_b,\mu,\beta)> 0\), there is a unique equilibrium where the sender chooses \(\hat{\pi}=(\hat{p},Q\hat{p})\), and \(V(\mu,\beta)>0\) is her equilibrium payoff. If \(\mathbf{F}(C_g,C_b,\mu,\beta)\leq 0\), \(f'(p)<0\) for all \(p\in(0,1)\). Hence, the sender's problem (\ref{eq.2}) has the corner solution \(\hat{p}=0\). In this case, the unique equilibrium outcome is the uninformative experiment.
\end{proof}

\begin{rmk}
	Notice that \(\ln Q+1-Q<0\), \(Q\ln Q+1-Q>0\) for all \(Q\in(0,1)\). Hence, Lemma \ref{lmm.a1} shows that persuasion is possible if and only if a weighted average of \(C_g\) and \(C_b\) is below some threshold.
\end{rmk}

\begin{lmm}\label{lmm.a7}
	Let \(C_g>\hat{K}(C_b)\). \(M(p,q), N(p,q)>0\) if \(q<p\leq\hat{\mathbf{p}}(q)\); \(M(p,q), N(p,q)<0\) if \(\check{\mathbf{p}}(q)\leq p<1\).
\end{lmm}

\begin{proof}
	Using the change of variable (\ref{eq.a17}), we can rewrite \(M(p,q)\) and \(N(p,q)\) as
	\begin{gather*}
		M(p,q) = 1+C_g(1-t+\ln t)=:m(t)\\
		N(p,q) = \frac{q}{p}\left[1-C_b\left(\frac{1}{t}-1+\ln t\right)\right]=:\frac{q}{p}\cdot n(t).
	\end{gather*}
	Both \(m(t)\) and \(n(t)\) are increasing in \(t\). Hence, it is sufficient to show that \(m(\hat{t}),n(\hat{t})>0\) and \(m(\check{t}),n(\check{t})<0\).

	Notice that
	\begin{equation}\label{eq.a22}
		\Delta(t) = -m(t) n(t) +\frac{1-t}{t}\big[C_b m(t)-tC_g n(t)\big],
	\end{equation}
	and
	\begin{equation}\label{eq.a23}
		\Delta'(t) = \frac{1}{t}\big[C_b m(t)-C_g n(t)\big].
	\end{equation}

	Recall that \(\Delta'(t^\star)=0\). Therefore, \(C_b m(t^\star)=C_g n(t^\star)\) by (\ref{eq.a23}), and
	\begin{equation*}
		m(t^\star) = \frac{C_b m(t^\star)-t^\star C_g n(t^\star)}{(1-t^\star)C_b} = \frac{t^\star}{(1-t^\star)^2C_b}\big[\Delta(t^\star)+m(t^\star) n(t^\star)\big].
	\end{equation*}
	Since \(\Delta(t^\star)>0\) and \(m(t^\star) n(t^\star)=\frac{C_b}{C_g}(m(t^\star))^2\geq 0\), we conclude that \(m(t^\star),n(t^\star)>0\). By monotonicity, \(m(\hat{t}),n(\hat{t})>0\).

	We now prove that \(m(\check{t}),n(\check{t})\) are negative. First, we show that \(m(\check{t}) n(\check{t})> 0\). Suppose that on the contrary, \(m(\check{t}) n(\check{t})\leq 0\). Since \(\Delta'(\check{t})>0\), (\ref{eq.a23}) implies that \(m(\check{t})\geq 0> n(\check{t})\) or \(m(\check{t})>0\geq n(\check{t})\). Hence, \(C_bm(\check{t})-\check{t}C_gn(\check{t})>0\). But \(\Delta(\check{t})=0\), so (\ref{eq.a22}) implies that \(m(\check{t}) n(\check{t})> 0\), a contradiction. Now suppose that \(m(\check{t}),n(\check{t})\) are both positive. Notice that \(\lim_{t\to 0}n(t)=-\infty\). Therefore, there exists a unique \(t'\in(0,\check{t})\) such that \(n(t')=0\). Since \(\Delta'\) is decreasing, \(\Delta'(t')>\Delta'(\check{t})>0\). Hence, by (\ref{eq.a23}), \(m(t')>0\), and by (\ref{eq.a22}), \(\Delta(t')>0=\Delta(\check{t})\). This is a contradiction to \(\Delta(t)\) being strictly increasing in \(t\) on \((0,t^\star)\) and that \(t'<\check{t}<t^\star\). Hence, \(m(\check{t}),n(\check{t})<0\).
\end{proof}

\begin{lmm}\label{lmm.a5}
	Let \(\mathbf{F}(C_g,C_b,\mu,\beta)>0\) and \(C_g>\hat{K}(C_b)\), and denote by \((\hat{p},Q\hat{p})\) the equilibrium outcome of the symmetric information benchmark. \(\hat{p}\in(\hat{\mathbf{p}}(q),\check{\mathbf{p}}(q))\), and \(\hat{p}\) is decreasing in \(\mu\).
\end{lmm}

\begin{proof}
	At \(\hat{p}\), the first order condition \(f'(p)=0\) holds. By (\ref{eq.b2}), \(M(\hat{p},Q\hat{p})\) and \(N(\hat{p},Q\hat{p})\) have different signs. It then follows from Lemma \ref{lmm.a7} that \(\hat{p}\in(\hat{\mathbf{p}}(q),\check{\mathbf{p}}(q))\). 

	By (\ref{eq.b2}),
	\begin{equation}\label{eq.e8}
		\frac{\partial}{\partial\mu}f'(\hat{p}) = M(\hat{p},Q\hat{p})-N(\hat{p},Q\hat{p}).
	\end{equation}
	Notice that the right-hand side of (\ref{eq.e8}) is decreasing in \(C_g\), hence fixing \(C_b\), it is negative if \(C_g\) is sufficiently large. Moreover, since \(M(\hat{p},Q\hat{p})\) and \(N(\hat{p},Q\hat{p})\) have different signs, (\ref{eq.e8}) is zero if and only if \(M(\hat{p},Q\hat{p})=N(\hat{p},Q\hat{p})=0\). This defines a system of equations of \(\hat{p}\) and \(C_g\). It is easy to verify that it has a unique solution
	\begin{equation*}
		\hat{p} = 1-\frac{1-Q}{Q}\frac{1}{\hat{x}(C_b)},~\text{and}~C_g=\hat{K}(C_b),
	\end{equation*}
	where \(\hat{x}(C_b)\) and \(\hat{K}(C_b)\) are given by (\ref{eq.b3}) and (\ref{eq.b4}). That is, if \(C_g>\hat{K}(C_b)\), (\ref{eq.e8}) is decreasing, hence \(\hat{p}\) is decreasing in \(\mu\).
\end{proof}

\subsection{Proof of Lemma \ref{lmm.a6}}

We first prove the following lemma, which states that the sender's expected payoff from choosing \((\hat{\mathbf{p}}(q),q)\) is increasing in \(q\) for all sender types. Two auxiliary functions \(J(p,q)\) and \(K(p,q)\) are introduced which are necessary for the proof of Lemma \ref{lmm.a6}.

\begin{lmm}\label{lmm.a8}
	Let \(C_g>\hat{K}(C_b)\). For all \(\mu\in(0,1)\), \(f((\hat{\mathbf{p}}(q),q),\mu)\) is positive and increasing in \(q\).
\end{lmm}

\begin{proof}
	By (\ref{eq.b5}), \(q=[q+\hat{t}(1-q)]\hat{\mathbf{p}}(q)>\hat{t}\hat{\mathbf{p}}(q)\), hence for all \(q\in(0,1)\)
	\begin{equation}\label{eq.a24}
		\hat{\mathbf{p}}(q)>f((\hat{\mathbf{p}}(q),q),\mu)>f((\hat{\mathbf{p}}(q),\hat{t}\hat{\mathbf{p}}(q)),\mu).
	\end{equation}
	Taking limit of (\ref{eq.a24}), we have \(\lim_{q\downarrow 0}f((\hat{\mathbf{p}}(q),q),\mu)=0\). Therefore, it suffices to show that \(f((\hat{\mathbf{p}}(q),q),\mu)\) is increasing in \(q\).

	By (\ref{eq.f1}) and (\ref{eq.f2}),
	\begin{equation*}
		\begin{aligned}
			\frac{\partial}{\partial q}f((\hat{\mathbf{p}}(q),q),\mu) &= \hat{\mathbf{p}}'(q)\left[\mu-\mu C_g\ln\frac{\hat{\mathbf{p}}(q)(1-q)}{(1-\hat{\mathbf{p}}(q))q}-(1-\mu) C_b\frac{\hat{\mathbf{p}}(q)-q}{\hat{\mathbf{p}}(q)(1-\hat{\mathbf{p}}(q))}\right] \\ 
			&\qquad\qquad\qquad + \left[(1-\mu)+\mu C_g\frac{\hat{\mathbf{p}}(q)-q}{q(1-q)}+(1-\mu) C_b\ln\frac{\hat{\mathbf{p}}(q)(1-q)}{(1-\hat{\mathbf{p}}(q))q}\right],
		\end{aligned}
	\end{equation*}
	and by (\ref{eq.b5}),
	\begin{equation*}
		\hat{\mathbf{p}}'(q) = \frac{\hat{\mathbf{p}}(q)(1-\hat{\mathbf{p}}(q))}{q(1-q)}.
	\end{equation*}
	Therefore,
	\begin{equation}\label{eq.a16}
		\frac{\partial}{\partial q}f((\hat{\mathbf{p}}(q),q),\mu) = J(\hat{\mathbf{p}}(q),q)\mu + K(\hat{\mathbf{p}}(q),q)(1-\mu),
	\end{equation}
	where
	\begin{gather*}
		J(p,q) := \frac{p(1-p)}{q(1-q)}\left(1+C_g\left[\frac{p-q}{p(1-p)}-\ln\frac{p(1-q)}{(1-p)q}\right]\right), \\
		K(p,q) := 1-C_b\left[\frac{p-q}{q(1-q)}-\ln\frac{p(1-q)}{(1-p)q}\right].
	\end{gather*}
	Substituting \(\hat{\mathbf{p}}(q)\) using (\ref{eq.b5}) and using the fact that \(\frac{q}{\hat{\mathbf{p}}(q)}>\hat{t}\), we have
	\begin{equation*}
		\frac{q(1-q)}{\hat{\mathbf{p}}(q)(1-\hat{\mathbf{p}}(q))}J(\hat{\mathbf{p}}(q),q) = 1+C_g\left[\frac{q}{\hat{\mathbf{p}}(q)}\frac{1-\hat{t}}{\hat{t}}+\ln\hat{t}\right] > 1+C_g(1-\hat{t}+\ln \hat{t}) = m(\hat{t}),
	\end{equation*}
	and
	\begin{equation*}
		K(\hat{\mathbf{p}}(q),q)=1-C_b\left[\frac{\hat{\mathbf{p}}(q)}{q}(1-\hat{t})+\ln\hat{t}\right] > 1-C_b\left(\frac{1}{\hat{t}}-1+\ln \hat{t}\right) = n(\hat{t}).
	\end{equation*}
	By Lemma \ref{lmm.a7}, \(m(\hat{t}), n(\hat{t})\) are positive. Hence, \(J(\hat{\mathbf{p}}(q),q), K(\hat{\mathbf{p}}(q),q)\) are positive, so \(f((\hat{\mathbf{p}}(q),q),\mu)\) is increasing in \(q\).
\end{proof}

We now prove Lemma \ref{lmm.a6}. To shorten notation, let \(\bar{Q}=\mathbf{Q}(\mu_N)\). 

Consider the type \(\theta\) sender's indifference curve through an experiment \(\pi=(\hat{\mathbf{p}}(q),q)\) such that \(\frac{q}{\hat{\mathbf{p}}(q)}\leq \bar{Q}\). This indifference curve intersects the straight line \(\frac{q}{p}=\bar{Q}\) twice. Let \((p_1,\bar{Q}p_1)\) and \((p_2,\bar{Q}p_2)\) be the intersections with \(p_1\leq\hat{\mathbf{p}}(q)<p_2\). That is, \(p_1\) and \(p_2\) solve
\begin{equation}\label{eq.a11}
	f((p,\bar{Q}p),\mu_\theta) = f((\hat{\mathbf{p}}(q),q),\mu_\theta).
\end{equation}
This defines \(p_1\) and \(p_2\) as functions of \(\mu_\theta\) and \(q\).

We first argue that \(\pi\) is robust to large deviations if and only if \(p_2(\mu_\theta,q)\) is weakly decreasing in \(\theta\). To show necessity, suppose by way of contradiction that \(\pi\) is robust to large deviations, but \(p_2(\mu_i,q)<p_2(\mu_j,q)\) for some \(i<j\). Denote by \(\mathfrak{q}_\theta\) the indifference curves of sender type \(\theta\) through \(\pi\). By Proposition \ref{pro3_part2}, \(\mathfrak{q}_i<\mathfrak{q}_j\) on a small neighborhood to the right of \(\hat{\mathbf{p}}(q)\), but \(\mathfrak{q}_j(p_2(\mu_i,q))<\bar{Q}p_2(\mu_i,q)=\mathfrak{q}_i(p_2(\mu_i,q))\). Therefore, there exists some \(\pi'=(p',q')\) such that \(p'\in(\hat{\mathbf{p}}(q),p_2(\mu_i,q))\) and \(\frac{q'}{p'}<\bar{Q}\) where all sender types' indifference curves intersect. Hence, we can find a deviation \(\tilde{\pi}=(\tilde{p},\tilde{q})\) such that \(\tilde{\pi}\) is persuasive at belief \(\mu_N\), \(f(\tilde{\pi},\mu_N)>f(\pi,\mu_N)\), and \(f(\tilde{\pi},\mu_\theta)<f(\pi,\mu_\theta)\) for all \(\theta<N\). That is, \(\pi\) is not robust to large deviations, a contradiction.

Conversely, suppose that \(p_2(\mu_\theta,q)\) is weakly decreasing in \(\theta\), but \(\pi\) is not robust to large deviations. That is, there exists an experiment \(\tilde{\pi}\) that is persuasive at belief \(\mu_N\), and \(i<N\) such that \(\mathfrak{q}_i(\tilde{p})>\tilde{q}>\mathfrak{q}_N(\tilde{p})\). By Proposition \ref{pro3_part2}, the two indifference curves \(\mathfrak{q}_i\) and \(\mathfrak{q}_N\) intersect at some experiment \((p',q')\) where \(p'\in(\hat{\mathbf{p}}(q),\tilde{p})\). By Proposition \ref{pro3_part2}, \(\mathfrak{q}_i>\mathfrak{q}_N\) on \((p',1)\). Specifically, at \(p_2(\mu_i,q)\), \(\mathfrak{q}_i(p_2(\mu_i,q))=\bar{Q}p_2(\mu_i,q)>\mathfrak{q}_N(p_2(\mu_i,q))\). Hence, \(p_2(\mu_i,q)<p_2(\mu_N,q)\). This is a contradiction to the assumption that \(p_2(\mu_\theta,q)\) is weakly decreasing in \(\theta\).

We now show that \(p_2(\mu_\theta,q)\) is weakly decreasing in \(\theta\) if and only if \(q\geq\hat{q}\) for some \(\hat{q}\). This is done by showing that the solution to the following system of equations of \(p\) and \(q\)
\begin{equation}\label{eq.b8}
	p_2(\mu,q) = p~\text{for all}~\mu\in[0,1]
\end{equation}
is unique.\footnote{Notice that \(f(\pi,\mu)\) is linear in \(\mu\) for all \(\pi\in\Pi\). Therefore, (\ref{eq.b8}) is equivalent to \(p_2(\mu_i,q)=p_2(\mu_j,q)=p\) for some \(i\neq j\). If (\ref{eq.b8}) does not have a solution, either \(p_2(\mu_\theta,q)\) is strictly decreasing in \(\theta\) for all \(q\), or it is never weakly decreasing. In the former case, \(\hat{q}=0\); in the latter case, we can without loss let \(\hat{q}=1\). Hence, we proceed assuming that (\ref{eq.b8}) has a solution.} Moreover, denoting by \((\hat{p},\hat{q})\) the solution of (\ref{eq.b8}), \(p_2(\mu,q)\) is strictly increasing in \(\mu\) if \(q<\hat{q}\), and it is strictly decreasing in \(\mu\) if \(q>\hat{q}\).

Let \((\hat{p},\hat{q})\) be a solution to (\ref{eq.b8}). It is sufficient to show that
\begin{equation}\label{eq.b6}
	\frac{\partial^2}{\partial\mu\partial q} p_2(\mu,q)\Big|_{q=\hat{q}} < 0.
\end{equation}
Evaluating the left-hand side of (\ref{eq.a11}) at \(p=p_2:=p_2(\mu,q)\) and taking derivative with respect to \(q\),
\begin{equation*}
	\frac{\partial}{\partial q} f((p_2,\bar{Q}p_2),\mu) = [M(p_2,\bar{Q}p_2)\mu+N(p_2,\bar{Q}p_2)(1-\mu)] \frac{\partial p_2}{\partial q}.
\end{equation*}
The derivative of the right-hand side of (\ref{eq.a11}) with respect to \(q\) is given already by (\ref{eq.a16}).

Therefore,
\begin{equation}\label{eq.b12}
	\frac{\partial p_2}{\partial q} = \frac{J(\hat{\mathbf{p}}(q),q)\mu+K(\hat{\mathbf{p}}(q),q)(1-\mu)}{M(p_2,\bar{Q}p_2)\mu+N(p_2,\bar{Q}p_2)(1-\mu)}.
\end{equation}
Taking derivative of (\ref{eq.b12}) with respect to \(\mu\), the left-hand side of (\ref{eq.b6}) has the same sign as the numerator
\begin{equation}\label{eq.f6}
	N(\hat{p},\bar{Q}\hat{p})J(\hat{\mathbf{p}}(\hat{q}),\hat{q})-M(\hat{p},\bar{Q}\hat{p})K(\hat{\mathbf{p}}(\hat{q}),\hat{q}).
\end{equation}
Hence, it is equivalent to show that (\ref{eq.f6}) is negative.

By Proposition \ref{pro3}, at any experiment \((p,q)\) such that \(\hat{\mathbf{p}}(q)<p<\check{\mathbf{p}}(q)\), the slope of the indifference curve of a higher type sender is greater than that of a lower type sender. Hence, if the indifference curves of two sender types intersect twice at \((\hat{q},\hat{\mathbf{p}}(\hat{q}))\) and at \((\hat{p},Q\hat{p})\), it must be the case that \(\hat{p}>\check{\mathbf{p}}(Q\hat{p})\). By Lemma \ref{lmm.a7}, \(M(\hat{p},Q\hat{p}), N(\hat{p},Q\hat{p})<0\). By Lemma \ref{lmm.a8}, \(J(\hat{\mathbf{p}}(\hat{q}),\hat{q}), K(\hat{\mathbf{p}}(\hat{q}),\hat{q})>0\). Lastly, we observe that
\begin{equation*}
	N(p,q)J(p,q)-M(p,q)K(p,q)=\Delta(p,q).
\end{equation*}
Hence, \(N(\hat{\mathbf{p}}(\hat{q}),\hat{q})J(\hat{\mathbf{p}}(\hat{q}),\hat{q})-M(\hat{\mathbf{p}}(\hat{q}),\hat{q})K(\hat{\mathbf{p}}(\hat{q}),\hat{q})=0\). Therefore, it is equivalent to show
\begin{equation*}
	\frac{N(\hat{p},\bar{Q}\hat{p})}{M(\hat{p},\bar{Q}\hat{p})} > \frac{K(\hat{\mathbf{p}}(\hat{q}),\hat{q})}{J(\hat{\mathbf{p}}(\hat{q}),\hat{q})} = \frac{N(\hat{\mathbf{p}}(\hat{q}),\hat{q})}{M(\hat{\mathbf{p}}(\hat{q}),\hat{q})}.
\end{equation*}

After a change of variables, we can rewrite \(\frac{N(p,q)}{M(p,q)}\) as a function of \(t\), as is defined in (\ref{eq.a17}), and \(Q=\frac{q}{p}\),
\begin{equation*}
	\frac{N(p,q)}{M(p,q)} = Q\underbrace{\frac{\left(1-C_b\left[\frac{1}{t}-1+\ln t\right]\right)}{1+C_g\left(1-t+\ln t\right)}}_{=:g(t)}.
\end{equation*}
Let \(\bar{t}=\frac{\bar{Q}(1-\hat{p})}{1-\bar{Q}\hat{p}}\) and \(\hat{Q}=\frac{\hat{q}}{\hat{\mathbf{p}}(\hat{q})}\). We need to show that \(\bar{Q}g(\bar{t})>\hat{Q}g(\hat{t})\). Since \(\hat{Q}<\bar{Q}\), and \(g\) is positive, it is sufficient to show \(g(\bar{t})>g(\hat{t})\). Since \(\hat{t}>\bar{t}\), the proof is complete once we show that \(g(t)\) is decreasing in \(t\).

Observe that
\begin{equation}\label{eq.b7}
	g'(t) = \frac{1-t}{t^2\left(1+C_g\left(1-t+\ln t\right)\right)^2}\big[C_b-C_g t+C_gC_b(2(1-t)+(1+t)\ln t)\big].
\end{equation}
Taking derivative of the bracket in (\ref{eq.b7}) with respect to \(t\) yields
\begin{equation*}
	C_g\left[C_b\left(\ln t+\frac{1}{t}-1\right)-1\right],
\end{equation*}
which is strictly decreasing in \(t\) on \((0,1)\) and has a unique zero \(\left(1+\hat{x}(C_b)\right)^{-1}\), where \(\hat{x}(C_b)\) is the unique solution to (\ref{eq.b4}) on the positive real line. That is, the bracket in (\ref{eq.b7}) obtains its maximum at \(t=\left(1+\hat{x}(C_b)\right)^{-1}\). Moreover, \(2(1-t)+(1+t)\ln t<0\) for all \(t\in(0,1)\). Hence, the bracket in (\ref{eq.b7}) is strictly decreasing in \(C_g\). Therefore, substituting \(t=\left(1+\hat{x}(C_b)\right)^{-1}\) and \(C_g=\hat{K}(C_b)\), we obtain a strict upper bound of the bracket in (\ref{eq.b7})
\begin{equation}\label{eq.b9}
 	C_b-\frac{\hat{K}(C_b)}{1+\hat{x}(C_b)}+C_b\frac{\hat{K}(C_b)}{1+\hat{x}(C_b)}\left[2\hat{x}(C_b)-\left(2+\hat{x}(C_b)\right)\ln(1+\hat{x}(C_b))\right].
\end{equation}
Substituting \(\hat{K}(C_b)\) using (\ref{eq.b3}) and \(\ln(1+\hat{x}(C_b))\) using (\ref{eq.b4}), (\ref{eq.b9}) becomes zero. That is, \(g'(t)<0\).

\subsection{Proof of Proposition \ref{pro5}}

The necessity of (iii) and (iv) is shown in the main text. We divide the rest of the proof into two lemmas. Lemma \ref{pro1} shows that an experiment satisfying (iii) is a pooling equilibrium outcome if and only if (i) and (ii) are satisfied. Lemma \ref{lmm.a6_part2} shows that a pooling equilibrium outcome \((\hat{\mathbf{p}}(q),q)\) satisfies D1 if (iv) is satisfied.

\begin{lmm}\label{pro1}
	An experiment \(\pi=(\hat{\mathbf{p}}(q),q)\) is a pooling equilibrium outcome if and only if:
	\begin{enumerate*}[label=(\roman*)]
		\item \(\frac{q}{\hat{\mathbf{p}}(q)}\leq\mathbf{Q}(\mu_0)\);
		\item \(f(\pi,\mu_1)\geq V(\mu_1,\mu_1)\).
	\end{enumerate*}
\end{lmm}

\begin{proof}
	\textit{Necessity.} On the equilibrium path, \(\beta(\pi)=\mu_0\). Hence, \(\pi\) is persuasive at \(\mu_0\). That is, (i) is satisfied. Consider a deviation \(\tilde{\pi}\) that is persuasive at belief \(\mu_1\). The lowest sender type gets \(f(\tilde{\pi},\mu_1)\) by deviating to \(\tilde{\pi}\) regardless of the receiver's belief. Therefore, she does not have a profitable deviation only if her equilibrium payoff \(f(\pi,\mu_1)\geq \sup_{\tilde{\pi}\in\tilde{\Pi}_1} f(\tilde{\pi},\mu_1)=V(\mu_1,\mu_1)\), where \(\tilde{\Pi}_1\) denotes the set of all experiments that are persuasive at belief \(\mu_1\).

	\vskip 1em
	\textit{Sufficiency.} Let \(\pi=(\hat{\mathbf{p}}(q),q)\) satisfy conditions (i) and (ii). We first show that \(f(\pi,\mu_\theta)\geq V(\mu_\theta,\mu_1)\) for all sender types \(\theta>1\). If \(V(\mu_\theta,\mu_1)=0\), this follows from Lemma \ref{lmm.a8}, so we proceed assuming that \(V(\mu_\theta,\mu_1)>0\). By Lemma \ref{lmm.a8}, there exists \(\bar{q}_\theta\in(0,1)\) such that \(f(\pi,\mu_\theta)\geq V(\mu_\theta,\mu_1)\) if and only if \(q\geq\bar{q}_\theta\). By Lemma \ref{lmm.a1}, there exists a unique \(\hat{p}_\theta>0\) such that \(f((\hat{p}_\theta,\mathbf{Q}(\mu_1)\hat{p}_\theta),\mu_\theta)=V(\mu_\theta,\mu_1)\). By Lemma \ref{lmm.a5}, there exists \(\hat{p}_1>\hat{p}_\theta\) such that \(f((\hat{p}_1,\mathbf{Q}(\mu_1)\hat{p}_1),\mu_1)=V(\mu_1,\mu_1)\). Since \(f((p,\mathbf{Q}(\mu_1)p),\mu_1)\) is increasing in \(p\) on \([0,\hat{p}_1]\), \(f((\hat{p}_\theta,\mathbf{Q}(\mu_1)\hat{p}_\theta),\mu_1)>0\). Therefore, there exists \(\bar{q}_1\in(0,1)\) such that \(f(\pi,\mu_1)\geq f((\hat{p}_\theta,\mathbf{Q}(\mu_1)\hat{p}_\theta),\mu_1)\) if and only if \(q\geq\bar{q}_1\). By our assumption that (ii) is satisfied, \(f(\pi,\mu_1)\geq V(\mu_1,\mu_1)>f((\hat{p}_\theta,\mathbf{Q}(\mu_1)\hat{p}_\theta),\mu_1)\), so \(q>\bar{q}_1\). We are done once we show that \(\bar{q}_1>\bar{q}_\theta\).

	By way of contradiction, suppose that \(\bar{q}_1\leq\bar{q}_\theta\). Denote by \(\mathfrak{q}_1\) and \(\mathfrak{q}_\theta\) type 1 and type \(\theta\) senders' indifference curves through \((\hat{p}_\theta,\mathbf{Q}(\mu_1)\hat{p}_\theta)\), respectively. By Lemma \ref{lmm.a5}, \(\hat{\mathbf{p}}(\mathbf{Q}(\mu_1)\hat{p}_\theta)<\hat{p}_\theta<\check{\mathbf{p}}(\mathbf{Q}(\mu_1)\hat{p}_\theta)\). Therefore, by Proposition \ref{pro3}, \(\mathfrak{q}_\theta(p')<\mathfrak{q}_1(p')\) for all \(p'\) such that \(\hat{\mathbf{p}}(\mathfrak{q}_1(p'))\leq p'<\hat{p}_\theta\). Specifically, at \(p'=\hat{\mathbf{p}}(\bar{q}_\theta)\), we have \(\bar{q}_\theta=\mathfrak{q}_\theta(\hat{\mathbf{p}}(\bar{q}_\theta))<\mathfrak{q}_1(\hat{\mathbf{p}}(\bar{q}_\theta))\). But since \(\bar{q}_1\leq\bar{q}_\theta\), by Lemma \ref{lmm.a8}, \(f((\hat{\mathbf{p}}(\bar{q}_\theta),\bar{q}_\theta),\mu_1)\geq f((\hat{\mathbf{p}}(\bar{q}_1),\bar{q}_1),\mu_1)=f((\hat{p}_\theta,\mathbf{Q}(\mu_1)\hat{p}_\theta),\mu_1)\), so \(\bar{q}_\theta\geq\mathfrak{q}_1(\hat{\mathbf{p}}(\bar{q}_\theta))\), a contradiction. Therefore, \(\bar{q}_1>\bar{q}_\theta\).

	We now show by construction that \(\pi\) is a pooling equilibrium outcome. Let \(\pi_\theta=\pi\) for all \(\theta\in\Theta\), \(\beta(\pi)=\mu_0\), and \(\beta(\tilde{\pi})=\mu_1\) for all \(\tilde{\pi}\neq\pi\). Moreover, let \(\hat{\beta}\) and \(\mathbf{a}\) be such that \(\hat{\beta}(\pi',s)=\mathbf{B}(\beta(\pi'),\pi',s)\) for all \(\pi'\in\Pi\) and \(s\in\{g,b\}\), and \(\mathbf{a}(\pi',s)=1\) if and only if \(\hat{\beta}(\pi',s)\geq\bar{\beta}\). To check that \((\{\pi_\theta\}_{\theta\in\Theta},\mathbf{a},\beta,\hat{\beta})\) is an equilibrium, we only need to show that any deviation \(\tilde{\pi}\neq\pi\) is not profitable for the sender. If \(\tilde{\pi}\) is persuasive at belief \(\mu_1\), the sender's payoff from deviating to \(\tilde{\pi}\) is at most \(V(\mu_\theta,\mu_1)\leq f(\pi,\mu_\theta)\). If \(\tilde{\pi}\) is unpersuasive at belief \(\mu_1\), the sender's deviation payoff is \(-c(\tilde{\pi}|\mu_\theta)\leq 0\). Therefore, the sender does not have a profitable deviation.
\end{proof}

\begin{lmm}\label{lmm.a6_part2}
	A pooling equilibrium outcome \(\pi=(\hat{\mathbf{p}}(q),q)\) is a D1 pooling equilibrium outcome if it is robust to large deviations.
\end{lmm}

\begin{proof}
	Recall from the proof of Lemma \ref{lmm.a6} that \((\hat{\mathbf{p}}(q),q)\) is robust to large deviations if and only if \(p_2(\mu_\theta,q)\) is weakly decreasing in \(\theta\). We are to show that if \(p_2(\mu_\theta,q)\) is weakly decreasing in \(\theta\), and \(f(\tilde{\pi},\mu_i)>f(\pi,\mu_i)\) for some \(i\) and \(\tilde{\pi}\) that is persuasive at belief \(\mu_N\), then \(f(\tilde{\pi},\mu_j)>f(\pi,\mu_j)\) for all \(j<i\). Hence, the critical off-path receiver beliefs that assign any deviation to the lowest type is consistent with the D1 criterion.

	Let \(\mathfrak{q}_\theta\) be the type \(\theta\) sender's indifference curve through \(\pi\). Notice that \(f(\tilde{\pi},\mu_i)>f(\pi,\mu_i)\) if and only if \(\tilde{p}\in(p_1(\mu_i,q),p_2(\mu_i,q))\) and \(\tilde{q}\in(\mathfrak{q}_i(\tilde{p}),\mathbf{Q}(\mu_N)\tilde{p}]\). Fix any \(j<i\). By Proposition \ref{pro3_part2}, \(\mathfrak{q}_i(\tilde{p})>\mathfrak{q}_j(\tilde{p})\) for all \(\tilde{p}<\hat{\mathbf{p}}(q)\). Moreover, \(\mathfrak{q}_i>\mathfrak{q}_j\) on a small neighborhood to the right of \(\hat{\mathbf{p}}(q)\) and at \(p_2(\mu_i,q)\). Therefore, \(\mathfrak{q}_i(\tilde{p})>\mathfrak{q}_j(\tilde{p})\) for all \(\tilde{p}\in(\hat{\mathbf{p}}(q),p_2(\mu_i,q))\). Otherwise, \(\mathfrak{q}_i=\mathfrak{q}_j\) has two solutions \(p'_1<p'_2\) on \((\hat{\mathbf{p}}(q),p_2(\mu_i,q))\), and \(\frac{\partial}{\partial\mu}MRS((p'_1,\mathfrak{q}_i(p'_1))|\mu)<0<\frac{\partial}{\partial\mu}MRS((p'_2,\mathfrak{q}_i(p'_2))|\mu)\), which is a contradiction to Proposition \ref{pro3_part2}. Therefore, for all \(\tilde{p}\in(p_1(\mu_i,q),p_2(\mu_i,q))\), \(\mathfrak{q}_j(\tilde{p})\leq \mathfrak{q}_i(\tilde{p})\). That is, \(f(\tilde{\pi},\mu_i)>f(\pi,\mu_i)\) implies \(f(\tilde{\pi},\mu_j)>f(\pi,\mu_j)\).
\end{proof}

\subsection{Proof of Proposition \ref{pro11}}

The proof uses the following lemma, which we prove at the end. It says that there is a unique experiment (aside from the uninformative experiment) that gives all types of the sender zero payoff.

\begin{lmm}\label{lmm.a9}
	Let \(C_g>\hat{K}(C_b)\). There exists a unique \(\pi\in\Pi^\circ\) such that \(f(\pi,\mu_\theta)=0\) for all \(\theta\). Moreover, \(\frac{\partial}{\partial\mu}MRS(\pi|\mu)<0\).
\end{lmm}

Let \(\pi^\star=(p^\star,q^\star)\) be the unique \(\pi\) in Lemma \ref{lmm.a9}, and \(\bar{\mu}=\mathbf{Q}^{-1}\left(\frac{q^\star}{p^\star}\right)\). We show that an uninformative D1 equilibrium exists if and only if conditions (i) and (ii) are satisfied.

Suppose that an uninformative D1 equilibrium exists, but, contrary to the claim, (i) is not satisfied. Then there exists an experiment \(\hat{\pi}_1\) that gives the lowest sender type a positive payoff \(V(\mu_1,\mu_1)\) regardless of the receiver's belief, hence it is a profitable deviation, a contradiction. Suppose that (ii) is not satisfied, i.e., \(\frac{q^\star}{p^\star}<\mathbf{Q}(\mu_N)\). Then in a neighborhood of \(\pi^\star\), all experiments are persuasive at belief \(\mu_N\). We can thus choose in there an experiment \(\tilde{\pi}\) such that deviating to \(\tilde{\pi}\) is profitable for the highest sender type but strictly equilibrium dominated for all other types of the sender. Hence, the uninformative equilibrium does not satisfy the D1 criterion, a contradiction. Therefore, conditions (i) and (ii) must be satisfied.

Conversely, suppose that (i) and (ii) are satisfied. By construction, an experiment \(\pi=(p,q)\) that is persuasive at belief \(\mu_N\) gives some sender type \(\theta\) a positive payoff, i.e.,  \(f(\pi,\mu_\theta)>0\), only if \(p<p^\star\). And by Lemma \ref{lmm.a9}, \(f(\pi,\mu_{\theta'})>0\) for all \(\theta'<\theta\). That is, the lowest sender type is the most keen to deviate. By (i), if the receiver's off-path belief is such that \(\beta(\pi)=\mu_1\) for all \(\pi\), the lowest sender type does not have a profitable deviation from choosing the uninformative experiment. Therefore, there exists an uninformative D1 equilibrium where \(\beta(\pi)=\mu_1\) for all \(\pi\) that is not the uninformative experiment.

\begin{proof}[\indent Proof of Lemma \ref{lmm.a9}]
	Since \(f(\pi,\mu)\) is an affine function in \(\mu\), \(f(\pi,\mu_\theta)=0\) for all \(\theta\) is equivalent to the following system of equations
	\begin{gather}
		p-C_g\left[p\ln\frac{p}{q}-(1-p)\ln\frac{1-q}{1-p}\right]=0, \label{eq.e3}\\
		q+C_b\left[q\ln\frac{p}{q}-(1-q)\ln\frac{1-q}{1-p}\right]=0. \label{eq.e4}
	\end{gather}
	
	The left-hand side of (\ref{eq.e3}) is increasing in \(q\). Given any \(p\in(0,1)\), it goes to negative infinity as \(q\) goes to zero and equals \(p\) when \(q=p\). Therefore, (\ref{eq.e3}) defines a function \(\sigma_1:p\mapsto q\) such that \(0<\sigma_1(p)<p\) for all \(p\in(0,1)\). Similarly, (\ref{eq.e4}) defines a function \(\sigma_0:p\mapsto q\) on \((0,1)\). Both \(\sigma_1\) and \(\sigma_0\) are continuous, hence their limits are well defined at 0 and 1. As \(p\) goes to 0, \(\lim_{p\downarrow 0}\sigma_1(p)=\lim_{p\downarrow 0}\sigma_0(p)=0\). As \(p\) goes to 1,
	\begin{equation}\label{eq.e5}
		\lim_{p\uparrow 1}\sigma_1(p) = \exp\left(-\frac{1}{C_g}\right)<1=\lim_{p\uparrow 1}\sigma_0(p).
	\end{equation}
	The first equality in (\ref{eq.e5}) is by observing that the left-hand side of (\ref{eq.e3}) converges to \(1+C_g\ln(q)\) as \(p\) goes to 1. The second equality in (\ref{eq.e5}) is by observing that, the limit resides in \([0,1]\), but for all \(q<1\), the left-hand side of (\ref{eq.e4}) goes to negative infinity as \(p\) goes to 1.

	We now show that \(\lim_{p\downarrow 0}\sigma_1'(p)>\lim_{p\downarrow 0}\sigma_0'(p)\). Therefore, for small positive \(p\), \(\sigma_1(p)>\sigma_0(p)\), so a solution exists for the system of equations (\ref{eq.e3}), (\ref{eq.e4}). Notice that, intuitively, \(\sigma_1\) is the sender's zero-payoff curve if her prior is 1. Therefore, \(V(1,\mu_r)>0\) if and only if \(\mathbf{Q}(\mu_r)>\lim_{p\downarrow 0}\sigma_1'(p)\). By Lemma \ref{lmm.a1}, \(Q=\lim_{p\downarrow 0}\sigma_1'(p)\) is the solution to
	\begin{equation}\label{eq.e6}
		1+C_g\left(\ln Q+1-Q\right) = 0.
	\end{equation}
	Notice that \(\ln Q+1-Q\) is increasing in \(Q\) on \((0,1]\). As \(Q\) goes to 0, it goes to negative infinity, and it equals 0 when \(Q=1\). Therefore, (\ref{eq.e6}) has a unique solution which is increasing in \(C_g\). By substituting \(C_g=\hat{K}(C_b)\) using (\ref{eq.b3}) and solving (\ref{eq.e6}), we obtain a strict lower bound
	\begin{equation*}
		\lim_{p\downarrow 0}\sigma_1'(p) > \frac{1}{1+\hat{x}(C_b)}.
	\end{equation*}
	Similarly, \(\sigma_0\) is the sender's zero-payoff curve if her prior is 0. Therefore, \(\lim_{p\downarrow 0}\sigma_0'(p)\) solves \(Q-C_b(Q\ln Q+1-Q)=0\). By (\ref{eq.b4}),
	\begin{equation*}
		\lim_{p\downarrow 0}\sigma_0'(p)=\frac{1}{1+\hat{x}(C_b)}.
	\end{equation*}

	Let \(\pi=(p,q)\) be a solution to (\ref{eq.e3}), (\ref{eq.e4}). We wan to show that \(\sigma_1'(p)<\sigma_0'(p)\). This simultaneously implies uniqueness and that \(\frac{\partial}{\partial\mu}MRS(\pi|\mu)<0\), since \(\sigma_1'(p)=MRS(\pi|1)\) and \(\sigma_0'(p)=MRS(\pi|0)\). Taking derivative of the left-hand side of (\ref{eq.e3}) with respect to \(p\) and evaluating at \(\pi\) yields
	\begin{equation*}
		1-C_g\left[\ln\frac{p}{q}+\ln\frac{1-q}{1-p}\right] = C_g\frac{1}{p}\ln\frac{1-q}{1-p}.
	\end{equation*}
	The derivative with respect to \(q\) is \(\frac{p-q}{q(1-q)}\). Hence,
	\begin{equation*}
		\sigma_1'(p) = \frac{q(1-q)}{p(p-q)}\ln\frac{1-q}{1-p}.
	\end{equation*}
	Similarly,
	\begin{equation*}
		\sigma_0'(p) = \frac{q(p-q)}{p(1-p)}\left(\ln\frac{1-q}{1-p}\right)^{-1}.
	\end{equation*}
	Therefore, \(\sigma_1'(p)<\sigma_0'(p)\) is equivalent to
	\begin{equation}\label{eq.e7}
		\left(\ln\frac{1-q}{1-p}\right)^2<\frac{(p-q)^2}{(1-p)(1-q)}.
	\end{equation}
	Notice that both sides of (\ref{eq.e7}) are decreasing in \(q\) on \((0,p]\), and they equal 0 at \(q=p\). Hence, it is sufficient to show that the right-hand decreases faster, that is,
	\begin{equation*}
		-2\ln\frac{1-q}{1-p} > -\frac{1-q}{1-p}+\frac{1-p}{1-q}.
	\end{equation*}
	This is true since \(\frac{1-q}{1-p}>1\) and \(2\ln x-x+\frac{1}{x}<0\) for all \(x>1\).
\end{proof}

\subsection{Proof of Proposition \ref{pro6}}

\begin{proof}
	We first show that there exists a unique collection of experiments satisfying conditions (i) to (iii). We then show that \(\{\pi_\theta\}_{\theta\in\Theta}\) is a D1 equilibrium outcome if and only if these conditions are satisfied.

	\vskip1em
	\textit{Existence and uniqueness.} Let \(\theta\geq 2\) be such that \(V(\mu_\theta,\mu_\theta)>0\), and fix \(v^\star_{\theta-1}\in[0,V(\mu_{\theta-1},\mu_{\theta-1})]\) such that \(v^\star_{\theta-1}>0\) if \(V(\mu_{\theta-1},\mu_{\theta-1})>0\). Consider the following maximization problem with a relaxed constraint
	\begin{equation}\label{eq.a5}
		\begin{aligned}
			&\max_{\pi=(p,q)} f(\pi,\mu_\theta) \\
			&s.t.~ \frac{q}{p}\leq\mathbf{Q}(\mu_\theta)~\text{and}~f(\pi,\mu_{\theta-1})\leq v_{\theta-1}^\star.
		\end{aligned}
	\end{equation}
	We are to show that (\ref{eq.a5}) has a unique solution \((p_\theta,q_\theta)\), and \(\frac{q_\theta}{p_\theta}=\mathbf{Q}(\mu_\theta)\). Therefore, (\ref{eq.a5}) is equivalent to (\ref{eq.5}). Moreover, \(v^\star_\theta\) is bounded from above by \(V(\mu_\theta,\mu_\theta)\) and, as we show later in the proof, strictly positive. Hence, by induction on \(\theta\), we can solve a unique collection of experiments \(\{\pi_\theta\}_{\theta\in\Theta}\) satisfying conditions (i) to (iii).

	If \(v^\star_{\theta-1}>0\), there exists a unique \(\tilde{\mu}\in(0,\mu_{\theta-1}]\) such that \(V(\mu_{\theta-1},\tilde{\mu})=v^\star_{\theta-1}\). If \(v^\star_{\theta-1}=0\), \(V(\mu_{\theta-1},\mu_{\theta-1})=0\). Therefore, \(\mathbf{F}(C_g,C_b,\mu_{\theta-1},\mu_{\theta-1})<0\), and by assumption, \(\mathbf{F}(C_g,C_b,\mu_\theta,\mu_\theta)>0\), where \(\mathbf{F}\) is defined in Lemma \ref{lmm.a1}. Hence, there exists some \(\tilde{\mu}\in[\mu_{\theta-1},\mu_\theta]\) such that \(\mathbf{F}(C_g,C_b,\mu_{\theta-1},\tilde{\mu})\leq 0\) and \(\mathbf{F}(C_g,C_b,\mu_\theta,\tilde{\mu})>0\), i.e., \(V(\mu_\theta,\tilde{\mu})>0\), and \(V(\mu_{\theta-1},\tilde{\mu})=0\). In either case, there exists a unique experiment \(\tilde{\pi}\) that is persuasive at belief \(\tilde{\mu}\) such that \(f(\tilde{\pi},\mu_\theta)=V(\mu_\theta,\tilde{\mu})\), and \(f(\tilde{\pi},\mu_{\theta-1})<V(\mu_{\theta-1},\tilde{\mu})=v^\star_{\theta-1}\). Hence, (\ref{eq.a5}) is equivalent to
	\begin{equation}\label{eq.d12}
		\max_{\pi\in\hat{\Pi}} f(\pi,\mu_\theta)
	\end{equation}
	where \(\hat{\Pi}=\{\pi=(p,q):\frac{q}{p}\leq\mathbf{Q}(\mu_\theta), f(\pi,\mu_{\theta-1})\leq v^\star_{\theta-1},f(\pi,\mu_\theta)\geq f(\tilde{\pi},\mu_\theta)\}\) is nonempty and compact. Moreover, \(\hat{\Pi}\) does not contain \((0,0)\), so \(f(\cdot,\mu_\theta)\) is continuous on \(\hat{\Pi}\). Therefore, (\ref{eq.d12}) admits a solution and its value is at least \(f(\tilde{\pi},\mu_\theta)>0\).

	Let \(\pi_\theta=(p_\theta,q_\theta)\) be a solution of (\ref{eq.d12}). Uniqueness is shown once we show that \(\frac{q_\theta}{p_\theta}=\mathbf{Q}(\mu_\theta)\), since by single-crossing, the intersection of \(\hat{\Pi}\) and \(\frac{q}{p}=\mathbf{Q}(\mu_\theta)\) is connected, and \(f((p,\mathbf{Q}(\mu_\theta)p),\mu_\theta)\) is concave in \(p\). Suppose by way of contradiction that \(\frac{q_\theta}{p_\theta}<\mathbf{Q}(\mu_\theta)\). Then in a neighborhood of \(\pi_\theta\), all experiments are persuasive at belief \(\mu_\theta\). By single-crossing, we can find in there an experiment \(\pi'\) such that \(f(\pi',\mu_\theta)>f(\pi_\theta,\mu_\theta)\geq f(\tilde{\pi},\mu_\theta)\) and \(f(\pi',\mu_{\theta-1})<f(\pi_\theta,\mu_{\theta-1})\leq v^\star_{\theta-1}\). That is, \(\pi'\in\hat{\Pi}\) and \(f(\pi',\mu_\theta)>f(\pi_\theta,\mu_\theta)\). This is a contradiction to \(\pi_\theta\) being a maximizer of (\ref{eq.d12}).

	\vskip1em
	\textit{Necessity.} Let \(\{\pi_\theta\}_{\theta\in\Theta}\) be a D1 equilibrium outcome. Assume by way of contradiction that \(V(\mu_\theta,\mu_\theta)=0\), but \(\pi_\theta\) is not the uninformative experiment for some sender type \(\theta\). By Lemma \ref{lmm4}, \(\pi_\theta\) is only chosen by the sender type \(\theta\), so it is persuasive at belief \(\beta(\pi_\theta)=\mu_\theta\). But \(V(\mu_\theta,\mu_\theta)=0\), so \(f(\pi_\theta,\mu_\theta)<0\). Therefore, the uninformative experiment is a profitable deviation for the sender, a contradiction. Hence, (i) is satisfied. Notice that the value of maximization problem (\ref{eq.d12}), which equals \(v^\star_\theta\) due to the equivalence between (\ref{eq.d12}) and (\ref{eq.5}), is at least \(f(\tilde{\pi},\mu_\theta)>0\). Therefore, (ii) is implied by (iii).

	We now show that \(v^\star_1=V(\mu_1,\mu_1)\). If \(V(\mu_1,\mu_1)=0\), \(v_1^\star=0\) by (i). If \(V(\mu_1,\mu_1)>0\), by Proposition \ref{pro3}, there exists a unique experiment \(\hat{\pi}_1\) that is persuasive at belief \(\mu_1\) such that \(f(\hat{\pi}_1,\mu_1)=V(\mu_1,\mu_1)\). That is, the lowest sender type gets \(V(\mu_1,\mu_1)\) from choosing \(\hat{\pi}_1\) regardless of the receiver's belief, so her equilibrium payoff \(v^\star_1\geq V(\mu_1,\mu_1)\). On the other hand, \(\pi_1\) is persuasive at belief \(\mu_1\), so \(v^\star_1=f(\pi_1,\mu_1)\leq V(\mu_1,\mu_1)\). Therefore, \(v^\star_1=V(\mu_1,\mu_1)\).

	For the rest of the proof, we denote by \(v^\star_\theta\) the unique payoff of the sender type \(\theta\) pinned down by (i) through (iii), and \(\pi^\star_{\theta}=(p^\star_{\theta-1},\mathbf{Q}(\mu_\theta)p^\star_\theta)\) the solution to (\ref{eq.d12}) for all \(\theta\geq 2\) such that \(V(\mu_\theta,\mu_\theta)>0\).

	In order to show that the rest of (iii) holds, we prove the following lemma: \(f(\pi^\star_\theta,\mu_{\theta'})\leq v^\star_{\theta'}\) for all \(\theta\geq 2\) such that \(V(\mu_\theta,\mu_\theta)>0\) and \(\theta'<\theta\). Consider two cases. First, if \(v^\star_{\theta-1}=0\), then by (i) and (ii), \(v^\star_{\theta'}=0\) for all sender types \(\theta'<\theta\). By single-crossing, \(f(\pi^\star_{\theta},\mu_{\theta-1})\leq 0\) implies that \(f(\pi^\star_\theta,\mu_{\theta'})\leq 0\) for all \(\theta'<\theta\). Second, if \(\pi^\star_{\theta-1}\) is not the uninformative experiment, by single-crossing, \(p>p^\star_{\theta-1}\) for all \(\pi=(p,\mathbf{Q}(\mu_\theta)p)\in\hat{\Pi}\). Otherwise, the indifference curve of the type \(\theta\) sender through \(\pi\) is lower than the indifference curve of the type \(\theta-1\) sender through \(\pi^\star_{\theta-1}\) on \([p,1)\). By construction, \(\tilde{\mu}\) is such the indifference curve of the type \(\theta-1\) sender through \(\pi^\star_{\theta-1}\) is tangent to the straight line \(\frac{q}{p}=\mathbf{Q}(\tilde{\mu})\). Therefore, for some \(\pi'=(p',q')\) such that \(f(\pi',\mu_\theta)=f(\pi,\mu_\theta)\), \(\frac{q'}{p'}<\mathbf{Q}(\tilde{\mu})\). Hence, \(f(\pi,\mu_{\theta})=f(\pi',\mu_\theta)<V(\mu_\theta,\tilde{\mu})=f(\tilde{\pi},\mu_\theta)\). That is, \(\pi\notin\hat{\Pi}\), a contradiction. Hence, \(p^\star_\theta>p^\star_{\theta-1}\). By single-crossing, \(f(\pi^\star_\theta,\mu_{\theta-1})\leq f(\pi^\star_{\theta-1},\mu_{\theta-1})\) implies that \(f(\pi^\star_\theta,\mu_{\theta'})\leq f(\pi^\star_{\theta-1},\mu_{\theta'})\) for all \(\theta'<\theta\). By induction on \(\theta\), we conclude that \(f(\pi^\star_\theta,\mu_{\theta'})\leq f(\pi^\star_{\theta'},\mu_{\theta'})=v^\star_{\theta'}\) for all \(\theta'<\theta\).

	We now show that condition (iii) holds. Suppose that \(\pi_\theta\neq\pi^\star_\theta\) for some \(\theta\geq 2\) such that \(V(\mu_\theta,\mu_\theta)>0\). That is, \(f(\pi^\star_\theta,\mu_\theta)>f(\pi_\theta,\mu_\theta)\). Let \(q'<q^\star_\theta\) be such that \(f(\pi',\mu_\theta)>f(\pi_\theta,\mu_\theta)\), where \(\pi'=(p^\star_\theta,q')\). \(\pi'\) is persuasive at belief \(\mu_\theta\), so it is a profitable deviation for the type \(\theta\) sender if \(\beta(\pi')\geq\mu_\theta\). Moreover, \(f(\pi',\mu_{\theta'})<f(\pi^\star_\theta,\mu_{\theta'})\leq v^\star_{\theta'}\) for all \(\theta'<\theta-1\). That is, \(\pi'\) is strictly equilibrium dominated for all sender types \(\theta'\leq\theta\). Hence, the D1 equilibrium requires that \(\beta(\pi')\geq\mu_\theta\), but given this off-path belief, \(\pi'\) is a profitable deviation for the type \(\theta\) sender, a contradiction. Hence, \(\pi_\theta=\pi^\star_\theta\) for all \(\theta\).

	\vskip1em
	\textit{Sufficiency.} Let \(\{\pi_\theta\}_{\theta\in\Theta}\) satisfy conditions (i) to (iii). We are to construct a D1 equilibrium with outcome \(\{\pi_\theta\}_{\theta\in\Theta}\). Let \(\beta(\pi_\theta)\) satisfy Bayes' rule for all \(\theta\in\Theta\), i.e., \(\beta(\pi_\theta)=\mu_\theta\) if \(\pi_\theta\) is not the uninformative experiment. For all off-the-equilibrium-path experiments \(\pi'\notin\{\pi_\theta\}_{\theta\in\Theta}\), let \(\hat{\Theta}({\pi'})=\{\theta\in\Theta:f(\pi',\mu_\theta)>v^\star_{\theta}\}\). That is, \(\hat{\Theta}({\pi'})\) is the set of sender types who may find \(\pi'\) to be a profitable deviation if \(\pi'\) is persuasive at belief \(\beta(\pi')\). Let \(\beta(\pi')=\mu_1\) if \(\hat{\Theta}(\pi')=\emptyset\) and \(\beta(\pi')=\min\hat{\Theta}(\pi')\) if \(\hat{\Theta}(\pi')\) is nonempty. Moreover, let \(\hat{\beta}\) and \(\mathbf{a}\) be such that \(\hat{\beta}(\pi',s)=\mathbf{B}(\beta(\pi'),\pi',s)\) for all \(\pi'\in\Pi\) and \(s\in\{g,b\}\), and \(\mathbf{a}(\pi',s)=1\) if and only if \(\hat{\beta}(\pi',s)\geq\bar{\beta}\). 

	Notice that \(D^0_\theta(\pi')=D_\theta(\pi')=[\mathbf{Q}^{-1}(\frac{q'}{p'}),\mu_N]\) for all \(\pi'=(p',q')\in\Pi^\circ\) and all \(\theta\in\hat{\Theta}(\pi')\), and \(D_\theta(\pi')=\emptyset\) for all \(\theta\notin\hat{\Theta}(\pi')\). Hence, the D1 criterion is satisfied. Hence, to check that \((\{\Pi_\theta\}_{\theta\in\Theta},\mathbf{a},\beta)\) is a D1 equilibrium, we only need to check that no sender type has a profitable deviation. If \(\beta(\pi')=\mu_1\), since all sender types' equilibrium payoffs \(v^\star_\theta\geq V(\mu_\theta,\mu_1)\), \(\pi'\) is not a profitable deviation. If \(\beta(\pi')=\mu_\theta\) for some \(\theta\geq 2\), \(\pi'\) is not persuasive at belief \(\mu_\theta\). Suppose by way of contradiction that \(\pi'\) is persuasive at belief \(\mu_\theta\). By definition, \(f(\pi',\mu_\theta)>f(\pi_\theta,\mu_\theta)\), \(f(\pi',\mu_{\theta-1})\leq v^\star_{\theta-1}\). This is a contradiction to that \(\pi_\theta\) is the solution to the maximization problem (\ref{eq.5}). Therefore, \(\pi'\) is not persuasive at belief \(\mu_\theta\). Hence, the payoff from deviating to \(\pi'\) is negative for all sender types, so there is no profitable deviation.
\end{proof}

\subsection{Proof of Proposition \ref{pro7}}

\begin{proof}
	If \(C_g>\hat{K}(C_b)\), in the symmetric information benchmark with no noisy signal, the sender chooses \(\pi^{si}(\mu_0)=(\hat{p},\mathbf{Q}(\mu_0)\hat{p})\), and \(\hat{p}>\hat{\mathbf{p}}(\mathbf{Q}(\mu_0)\hat{p}))\) by Lemma \ref{lmm.a5}. That is, \(\frac{\mathbf{Q}(\mu_0)\hat{p}}{\hat{\mathbf{p}}(\mathbf{Q}(\mu_0)\hat{p})}>\mathbf{Q}(\mu_0)\). Recall that \(\frac{\bar{q}}{\hat{\mathbf{p}}(\bar{q})}=\mathbf{Q}(\mu_0)\), and notice that \(\frac{q}{\hat{\mathbf{p}}(q)}=q+\hat{t}(1-q)\) is increasing in \(q\). Therefore, \(\mathbf{Q}(\mu_0)\hat{p}>\bar{q}\). Hence, \(\pi^{si}(\mu_0)\) is Blackwell more informative than the experiment \((\hat{\mathbf{p}}(\bar{q}),\bar{q})\).

	If \(C_g\leq\hat{K}(C_b)\), by Proposition \ref{pro6}, the sender chooses the same experiment \(\pi^{si}(\mu_\theta)\) in the D1 separating equilibrium if \(\theta=1\) or if \(V(\mu_\theta,\mu_\theta)=0\). If \(V(\mu_\theta,\mu_\theta)>0\) for some \(\theta>1\), denote by \(\pi^\star_\theta=(p^\star_\theta,\mathbf{Q}(\mu_\theta)p^\star_\theta)\) the experiment chosen by the type \(\theta\) sender in the D1 separating equilibrium. If \(\hat{\pi}_\theta\in\hat{\Pi}\), where \(\hat{\Pi}\) is defined in (\ref{eq.d12}), \(\pi^\star_\theta=\hat{\pi}_\theta\). If \(\hat{\pi}_\theta\notin\hat{\Pi}\), \(f((p,\mathbf{Q}(\mu_\theta)),\mu_\theta)\) is decreasing in \(p\) at \(p^\star_\theta\). Hence, \(p^\star_\theta>\hat{p}_\theta\), that is, \(\pi^\star_\theta\) is Blackwell more informative than \(\pi^{si}(\mu_\theta)\).
\end{proof}

\subsection{Results relating to more precise sender private information}

We assume in this section that \(\mu_N>\bar{\beta}\). We prove two lemmas. Lemma \ref{lmm.b8} shows that the sender's payoff \(1-c(\pi|\mu)\) satisfies the single-crossing property. Lemma \ref{lmm.b9} generalizes Lemma \ref{lmm4}. It states that pooling is only possible at an experiment \(\pi\in\Pi^\circ\) if sender types have the same marginal rate of substitution, and the receiver chooses the high (low) action after seeing the good (bad) outcome from \(\pi\).

\begin{lmm}\label{lmm.b8}
	For all \(\pi\in\Pi^\circ\),
	\begin{equation}\label{eq.e9}
		\frac{\partial}{\partial\mu}\left[-\frac{\partial c(\pi|\mu)/\partial p}{\partial c(\pi|\mu)\partial q}\right] < 0.
	\end{equation}
\end{lmm}

\begin{proof}
	Notice that
	\begin{equation*}
		-\frac{\partial c(\pi|\mu)/\partial p}{\partial c(\pi|\mu)\partial q} = \frac{A_1(p,q) + [A_2(p,q)+1]\mu}{[A_3(p,q)-1] + [A_4(p,q)+1]\mu},
	\end{equation*}
	where \(A_1\), \(A_2\), \(A_3\), \(A_4\) are defined in (\ref{eq.f1}) and (\ref{eq.f2}). Therefore, the left-hand side of (\ref{eq.e9}) equals
	\begin{equation}\label{eq.e10}
		\frac{[A_2(p,q)+1][A_3(p,q)-1]-A_1(p,q)[A_4(p,q)+1]}{\left([A_3(p,q)-1] + [A_4(p,q)+1]\mu\right)^2},
	\end{equation}
	which has the same sign as its numerator. Using the change of variable (\ref{eq.a17}), the numerator of (\ref{eq.e10}) equals
	\begin{equation*}
		C_gC_b\left[(\ln t)^2-\frac{(1-t)^2}{t}\right],
	\end{equation*}
	which is negative for all \(t\in(0,1)\). Hence, the left-hand side of (\ref{eq.e9}) is negative.
\end{proof}

\begin{lmm}\label{lmm.b9}
	Let \(\{\pi_\theta\}_{\theta\in\Theta}\) be a D1 equilibrium outcome, and \(\pi=(p,q)\in\Pi^\circ\). If \(\pi_i=\pi_j=\pi\) for some \(i,j\in\Theta\), \(MRS(\pi|\mu_i)=MRS(\pi|\mu_j)\), and \(\beta(\pi)<\mathbf{Q}^{-1}\left(\frac{1-q}{1-p}\right)\).
\end{lmm}

\begin{proof}
	Without loss of generality, let \(i\) be the highest sender type that chooses \(\pi\), hence \(\beta(\pi)<\mu_i\). We first show that \(\beta(\pi)<\mathbf{Q}^{-1}\left(\frac{1-q}{1-p}\right)\). Suppose by way of contradiction that \(\beta(\pi)\geq \mathbf{Q}^{-1}\left(\frac{1-q}{1-p}\right)\). Then in a small neighborhood of \(\pi\), all experiments \(\tilde{\pi}=(\tilde{p},\tilde{q})\) satisfy \(\mathbf{Q}^{-1}\left(\frac{1-\tilde{q}}{1-\tilde{p}}\right)<\mu_i\). By Lemma \ref{lmm.b8}, we can find in this neighborhood an experiment \(\tilde{\pi}\) such that \(1-c(\tilde{\pi}|\mu_i)>1-c(\pi|\mu_i)\), and \(1-c(\tilde{\pi}|\mu_k)>1-c(\pi|\mu_k)\) for all \(k<i\). Notice that the sender takes the high action if \(\pi\) is conducted regardless of its outcome, so the sender gets \(1-c(\pi|\mu_\theta)\) from choosing \(\pi\). Therefore, \(\tilde{\pi}\) is a profitable deviation for the type \(i\) sender if the receiver's interim belief \(\beta(\tilde{\pi})\geq\mathbf{Q}^{-1}\left(\frac{1-\tilde{q}}{1-\tilde{p}}\right)\), but it is strictly equilibrium dominated for all sender types \(k<i\). The D1 criterion thus requires that \(\beta(\tilde{\pi})\geq\mu_i>\mathbf{Q}^{-1}\left(\frac{1-\tilde{q}}{1-\tilde{p}}\right)\). This makes \(\tilde{\pi}\) a profitable deviation for the sender type \(i\), a contradiction.

	Therefore, \(\beta(\pi)<\mathbf{Q}^{-1}\left(\frac{1-q}{1-p}\right)\), and the sender's payoff from choosing \(\pi\) is \(f(\pi,\mu_\theta)\). The rest of the proof is in analogy to that of Lemma \ref{lmm4}. There exists a neighborhood of \(\pi\) wherein all experiments \(\tilde{\pi}\) are persuasive at belief \(\mu_i\) and \(1-c(\tilde{\pi}|\mu_\theta)>f(\pi,\mu_\theta)\). Suppose by way of contradiction that \(MRS(\pi|\mu_i)\neq MRS(\pi|\mu_j)\). Then there exists in this neighborhood an experiment \(\tilde{\pi}\) such that \(f(\tilde{\pi},\mu_i)>f(\pi,\mu_i)\), and \(f(\tilde{\pi},\mu_k)<f(\pi,\mu_k)\) for all \(k<i\). Deviating to \(\tilde{\pi}\) is strictly profitable for the sender type \(i\) if \(\beta(\tilde{\pi})\geq\mathbf{Q}^{-1}\left(\frac{\tilde{q}}{\tilde{p}}\right)\), but it gives all sender types \(k<i\) less than their equilibrium payoff if \(\beta(\tilde{\pi})\leq\mathbf{Q}^{-1}\left(\frac{1-\tilde{q}}{1-\tilde{p}}\right)\). Notice that \(\frac{1-\tilde{q}}{1-\tilde{p}}>1>\frac{\tilde{q}}{\tilde{p}}\), and \(\mathbf{Q}\) is increasing. Hence, the D1 criterion requires that \(\beta(\tilde{\pi})\geq\mu_i\), which makes \(\tilde{\pi}\) a profitable deviation for the sender type \(i\), a contradiction. Therefore, \(MRS(\pi|\mu_i)=MRS(\pi|\mu_j)\).
\end{proof}

\subsection{Proof of Proposition \ref{pro9}}

\begin{proof}
	For \(\underline{n}\leq n\leq\bar{n}\), let \(z_n\) be the expected cost of acquiring additional signals before the difference between the number of \(g\)'s and \(b\)'s reaches either threshold, conditional on the state being good and the current difference being \(n\). \(\{z_n\}_{\underline{n}\leq n\leq\bar{n}}\) satisfies the recurrence relation
	\begin{equation}\label{eq.a8}
		z_n = \bar{c}_g + \alpha z_{n+1} + (1-\alpha) z_{n-1}
	\end{equation}
	for all \(\underline{n}<n<\bar{n}\), where \(\bar{c}_g=\alpha c_g+(1-\alpha)c_b\), and the boundary conditions \(z_{\underline{n}}=z_{\bar{n}}=0\). Rewriting (\ref{eq.a8}) as
	\begin{equation*}
		\alpha\left(z_{n+1}+\frac{n+1}{2\alpha-1}\bar{c}_g\right)-\left(z_n+\frac{n}{2\alpha-1}\bar{c}_g\right)+(1-\alpha)\left(z_{n-1}+\frac{n-1}{2\alpha-1}\bar{c}_g\right)=0,
	\end{equation*}
	we can solve that
	\begin{equation*}
		z_n = C_1 x^n + C_2 - \frac{n}{2\alpha-1}\bar{c}_g,
	\end{equation*}
	where \(x=\frac{1-\alpha}{\alpha}\), and
	\begin{equation*}
		C_1=-\frac{\bar{c}_g}{2\alpha-1}\frac{\bar{n}-\underline{n}}{x^{\underline{n}}-x^{\bar{n}}},\quad C_2=\frac{\bar{c}_g}{2\alpha-1}\frac{\bar{n}x^{\underline{n}}-\underline{n}x^{\bar{n}}}{x^{\underline{n}}-x^{\bar{n}}}
	\end{equation*}
	are constants solved using the boundary conditions. Hence, the expected cost of implementing the threshold stopping rule \(\tau\) conditional on the good state, \(\mathbb{E}[c_gn_g(h_\tau)+c_bn_b(h_\tau)|\omega=G]\), equals
	\begin{equation*}
		z_0=\frac{\bar{c}_g}{2\alpha-1}\frac{\bar{n}(x^{\underline{n}}-1)-\underline{n}(x^{\bar{n}}-1)}{x^{\underline{n}}-x^{\bar{n}}}.
	\end{equation*}

	Similarly, conditional on the bad state, the expected cost of implementing \(\tau\),
	\begin{equation*}
		\mathbb{E}[c_gn_g(h_\tau)+c_bn_b(h_\tau)|\omega=B] = \frac{\bar{c}_b}{2\alpha-1}\frac{\bar{n}(x^{\bar{n}}-x^{\bar{n}+\underline{n}})-\underline{n}(x^{\underline{n}}-x^{\bar{n}+\underline{n}})}{x^{\underline{n}}-x^{\bar{n}}}.
	\end{equation*}
	Hence, the expected cost of implementing the strategy \(\tau\), \(\mathbb{E}[c_g n_g(h_\tau)+c_b n_b(h_\tau)]\), is
	\begin{equation}\label{eq.f3}
		\frac{1}{2\alpha-1}\left[\mu_0\frac{\bar{n}(x^{\underline{n}}-1)-\underline{n}(x^{\bar{n}}-1)}{x^{\underline{n}}-x^{\bar{n}}}\bar{c}_g+(1-\mu_0)\frac{\bar{n}(x^{\bar{n}}-x^{\bar{n}+\underline{n}})-\underline{n}(x^{\underline{n}}-x^{\bar{n}+\underline{n}})}{x^{\underline{n}}-x^{\bar{n}}}\bar{c}_b\right].
	\end{equation}

	We are to verify that (\ref{eq.f3}) equals \(\mathbb{E}[H(\mu_0)-H(\hat{\mu})]\). First, notice that
	\begin{equation*}
		\ln\left(\frac{1-\mu_n}{\mu_n}\right) = \ln\left(\frac{1-\mu_0}{\mu_0}\right)+n\ln x
	\end{equation*}
	for \(n=\bar{n},\underline{n}\). And
	\begin{gather*}
		\frac{\mu_0-\mu_{\underline{n}}}{\mu_{\bar{n}}-\mu_{\underline{n}}}\mu_{\bar{n}} = \frac{\mu_0\left(1+\frac{1-\mu_0}{\mu_0}x^{\underline{n}}\right)-1}{\left(1+\frac{1-\mu_0}{\mu_0}x^{\underline{n}}\right)-\left(1+\frac{1-\mu_0}{\mu_0}x^{\bar{n}}\right)} = \frac{\mu_0(x^{\underline{n}}-1)}{x^{\underline{n}}-x^{\bar{n}}},\\
		\frac{\mu_{\bar{n}}-\mu_0}{\mu_{\bar{n}}-\mu_{\underline{n}}}\mu_{\underline{n}} = \frac{1-\mu_0\left(1+\frac{1-\mu_0}{\mu_0}x^{\bar{n}}\right)}{\left(1+\frac{1-\mu_0}{\mu_0}x^{\underline{n}}\right)-\left(1+\frac{1-\mu_0}{\mu_0}x^{\bar{n}}\right)} = \frac{\mu_0(1-x^{\bar{n}})}{x^{\underline{n}}-x^{\bar{n}}}.
	\end{gather*}
	Hence,
	\begin{equation*}
		\mathbb{E}\left[\hat{\mu}\ln\left(\frac{1-\hat{\mu}}{\hat{\mu}}\right)\right] = \mu_0\ln\left(\frac{1-\mu_0}{\mu_0}\right) + \mu_0\frac{\bar{n}(x^{\underline{n}}-1)-\underline{n}(x^{\bar{n}}-1)}{x^{\underline{n}}-x^{\bar{n}}}\ln x.
	\end{equation*}
	Similarly,
	\begin{equation*}
		\mathbb{E}\left[(1-\hat{\mu})\ln\left(\frac{\hat{\mu}}{1-\hat{\mu}}\right)\right] = (1-\mu_0)\ln\left(\frac{\mu_0}{1-\mu_0}\right) + (1-\mu_0)\frac{\bar{n}(x^{\bar{n}}-x^{\bar{n}+\underline{n}})-\underline{n}(x^{\underline{n}}-x^{\bar{n}+\underline{n}})}{x^{\underline{n}}-x^{\bar{n}}}\ln x.
	\end{equation*}
	Therefore, (\ref{eq.f3}) equals \(\mathbb{E}[H(\mu_0)-H(\hat{\mu})]\).
\end{proof}

\subsection{Proof of Proposition \ref{pro10}}

\begin{proof}
	Let \(\pi=(p,q)\in\Pi^\circ\). In the game with Shannon Entropy cost,
	\begin{gather*}
		\frac{\partial f(\pi,\mu)}{\partial p} = \mu \left[1-C\left(\ln\frac{1-r}{r}-\ln\frac{1-p}{p}\right)\right], \\
		\frac{\partial f(\pi,\mu)}{\partial q} = (1-\mu) \left[1+C\left(\ln\frac{1-q}{q}-\ln\frac{1-r}{r}\right)\right],
	\end{gather*}
	where \(r:=\mu p+(1-\mu) q\). The marginal rate of substitution is therefore
	\begin{equation}\label{eq.a9}
		MRS(\pi|\mu) = -\frac{\mu}{1-\mu}\left[1-\frac{CL}{R}\right],
	\end{equation}
	where \(L:=\ln\frac{1-q}{q}-\ln\frac{1-p}{p}>0\), and \(R:=1+C\left(\ln\frac{1-q}{q}-\ln\frac{1-r}{r}\right)>1\).

	Taking derivative of (\ref{eq.a9}) with respect to \(\mu\),
	\begin{equation}\label{eq.d1}
		\frac{\partial}{\partial\mu} MRS(\pi|\mu) = \frac{MRS(\pi|\mu)}{\mu(1-\mu)} - \frac{\mu}{1-\mu}\frac{C^2L}{R^2}\frac{p-q}{r(1-r)}.
	\end{equation}
	Hence, if \(MRS(\pi|\mu)\leq 0\), \(\frac{\partial}{\partial\mu}MRS(\pi|\mu)<0\). Therefore, \(MRS(\pi|\mu)\) as a function of \(\mu\) has at most one zero on \((0,1)\). Moreover, notice that \(\lim_{\mu\downarrow 0}MRS(\pi|\mu)=0\), and \(\lim_{\mu\uparrow 1}MRS(\pi|\mu)=-\infty\). Hence, \(MRS(\pi|\mu)\) has a unique zero on \((0,1)\) if and only if
	\begin{equation*}
		\lim_{\mu\downarrow 0}\frac{\partial}{\partial\mu}MRS(\pi|\mu) = C\left(\ln\frac{1-q}{q}-\ln\frac{1-p}{p}\right)-1>0.
	\end{equation*}
	That is,
	\begin{equation*}
		p > \tilde{\mathbf{p}}(q) := \frac{q}{q+(1-q)e^\frac{1}{C}}.
	\end{equation*}
	If \(p\leq\tilde{\mathbf{p}}(q)\), \(MRS(\pi|\mu)\) is negative and decreasing in \(\mu\) for all \(\mu\in(0,1)\).

	We now show that if \(p>\tilde{\mathbf{p}}(q)\), \(MRS(\pi|\mu)\) is single-peaked. Hence, it first increases and then decreases in \(\mu\). This concludes statement (i). To shorten notation, we will briefly denote by \(M,M',M''\) the marginal rate of substitution \(MRS(\pi|\mu)\) and its first and second order derivatives with respect to \(\mu\), respectively. Multiplying (\ref{eq.d1}) with \(1-\mu\) and taking derivative with respect to \(\mu\),
	\begin{align}
		(1-\mu)M''-M' &= \frac{M'}{\mu}-\frac{M}{\mu^2} -\frac{C^2L}{R^2}\frac{p-q}{r^2(1-r)^2}(r^2-2rq+q) +2\frac{C^2L}{R^3}\frac{p-q}{r^2(1-r)^2}(r-q) \nonumber\\
		&< \frac{M'}{\mu}-\frac{M}{\mu^2}-\frac{C^2L}{R^2}\frac{p-q}{r^2(1-r)^2}(r^2-2rq+3q-2r) \label{eq.d2}
	\end{align}
	Suppose that at some \(\mu\in(0,1)\), \(M'=0\). By (\ref{eq.d1}), \(M>0\), and
	\begin{equation}\label{eq.d3}
		\frac{C^2L}{R^2}\frac{p-q}{r(1-r)} = \frac{M}{\mu^2}.
	\end{equation}
	Substituting \(M'=0\) and (\ref{eq.d3}) in (\ref{eq.d2}),
	\begin{equation*}
		(1-\mu)M'' < -\frac{M}{\mu^2}\frac{3q-r-2rq}{r(1-r)}.
	\end{equation*}
	Notice that
	\begin{equation}\label{eq.d4}
		3q-r-2rq = -2(1-\mu)q^2 + (4-\mu-2\mu p)q + \mu p.
	\end{equation}
	The right-hand side of (\ref{eq.d4}) is increasing in \(q\) on \((0,1)\), and it equals \(\mu p>0\) at \(q=0\). Hence, \(3q-r-2rq>0\), and \(M''<0\) at \(\mu\). That is, \(MRS(\pi|\mu)\) obtains its global maximum at \(\mu\), and it has no local minimum. Therefore, \(MRS(\pi|\mu)\) is single-peaked.
	
	Fix any \(q\in(0,1)\) and \(i<j\). Recall that, if \(p\leq\tilde{\mathbf{p}}(q)\), \(MRS(\pi|\mu)\) is decreasing in \(\mu\), and observe from (\ref{eq.d1}) that \(\lim_{p\uparrow 1}\frac{\partial}{\partial\mu}MRS(\pi|\mu)=\infty\) for all \(\mu\in(0,1)\). Hence, there exists some \(\tilde{p}>\tilde{\mathbf{p}}(q)\) such that \(MRS((\tilde{p},q)|\mu_i)=MRS((\tilde{p},q)|\mu_j)\). We are to show that
	\begin{equation*}
		\frac{\partial}{\partial p}\big(MRS((p,q)|\mu_i)-MRS((p,q)|\mu_j)\big)\Big|_{p=\tilde{p}} < 0.
	\end{equation*}
	Letting \(\tilde{\mathbf{p}}_{i,j}:q\mapsto\tilde{p}\), statement (ii) then follows.

	Taking derivative of (\ref{eq.a9}) with respect to \(p\),
	\begin{equation}\label{eq.d6}
		\frac{\partial}{\partial p}MRS(\pi|\mu) = \frac{\mu}{1-\mu}\frac{C}{R}\left(\frac{1}{p(1-p)}-\frac{CL}{R}\frac{\mu}{r(1-r)}\right).
	\end{equation}
	Substituting \(R\) in (\ref{eq.d6}) using (\ref{eq.a9}), we have
	\begin{equation*}
		\frac{\partial}{\partial p}MRS(\pi|\mu) = \frac{1}{L}\left(M+\frac{\mu}{1-\mu}\right)\left[\frac{1}{p(1-p)}-\left(1+\frac{1-\mu}{\mu}M\right)\frac{\mu}{r(1-r)}\right],
	\end{equation*}
	where we again use \(M\) to denote \(MRS(\pi|\mu)\). Notice that \(L>0\) is independent of \(\mu\), and \(MRS((\hat{p},q)|\mu_i)=MRS((\hat{p},q)|\mu_j)\). Hence, it is sufficient to show that
	\begin{equation}\label{eq.d7}
		G(\mu) := \left(m+\frac{\mu}{1-\mu}\right)\left[\frac{1}{p(1-p)}-\left(1+\frac{1-\mu}{\mu}m\right)\frac{\mu}{r(1-r)}\right]
	\end{equation}
	is increasing given any \(m, p\), and \(q\). Taking derivative of (\ref{eq.d7}) with respect to \(\mu\),
	\begin{align*}
		G'(\mu) &= \frac{r^2-2rp+p}{r^2(1-r)^2}m^2 - 2\frac{r^2-2rq+q}{r^2(1-r)^2}m \\
		&\qquad\qquad\qquad\qquad\qquad+ \frac{1}{(1-\mu)^2}\left[\frac{1}{p(1-p)}-\frac{\mu}{r(1-r)}\right]-\frac{\mu}{1-\mu}\frac{r^2-2rq+q}{r^2(1-r)^2} \\
		&= \frac{(mp-q)^2}{pr^2} + \frac{(m(1-p)-(1-q))^2}{(1-p)(1-r)^2} > 0.
	\end{align*}
	That is, \(G(\mu)\) is increasing.
\end{proof}

\bibliography{Signaling}

\end{document}